\documentclass[12pt]{article}
\usepackage{titling}
\usepackage{amsmath,amsthm,amsfonts}
\usepackage{graphicx,psfrag,epsf}
\usepackage{color}
\usepackage{enumerate}
\usepackage{natbib}
\usepackage[colorlinks=true,citecolor=blue,linkcolor=blue]{hyperref}
\usepackage{float}
\usepackage[version=3]{mhchem}
\usepackage[utf8]{inputenc}
\usepackage{subcaption}
\usepackage{mathtools}
\usepackage{mathrsfs}

\usepackage{enumitem}
\newcommand{\blind}{1}

\usepackage{xcolor}

\def\trans{^{\rm T}}

\newcommand{\bds}{{\bf s}}
\newcommand{\bdx}{{\bf x}}

\newcommand{\bdi}{{\bf i}}
\newcommand{\bdh}{{\bf h}}

\newcommand{\bdone}{{\bf 1}}
\newcommand{\BI}{{\bf I}}

\newcommand{\bdgamma}{{\pmb \gamma}}
\newcommand{\bdnu}{{\pmb \nu}}
\newcommand{\bdxi}{{\pmb \xi}}
\newcommand{\bdu}{{\pmb u}}
\newcommand{\var}{{\rm Var}}
\newcommand{\cov}{{\rm Cov}}
\newcommand{\E}{{\rm E}}
\newcommand{\bdbeta}{{\pmb \beta}}

\newcommand{\bdepsilon}{{\pmb \epsilon}}
\newcommand{\BX}{{\bf X}}
\newcommand{\BY}{{\bf Y}}
\newcommand{\BU}{{\bf U}}
\newcommand{\BSigma}{{\bf \Sigma}}

\def\wt{\widetilde}

\def\wh{\widehat}
\theoremstyle{plain}
\newtheorem{thm}{Theorem}[section]
\newtheorem{lm}{Lemma}[section]

\pagebreak
\allowdisplaybreaks

\begin{document}

\def\spacingset#1{\renewcommand{\baselinestretch}%
{#1}\small\normalsize} \spacingset{1}

%%%%%%%%%%%%%%%%%%%%%%%%%%%%%%%%%%%%%%%%%%%%%%%%%%%%%%%%%%%%%%%%%%%%%%%%%%%%%%

\if1\blind
{
  \title{\bf A Geospatial Functional Model For OCO-2 Data with Application on Imputation and Land Fraction Estimation}
  \author{Xinyue Chang$^1$, Zhengyuan Zhu$^1$, Xiongtao Dai$^1$ \\
  and Jonathan Hobbs$^2$ \\
  %\thanks{The authors gratefully acknowledge \textit{please remember to list all relevant funding sources in the unblinded version}}
    %\hspace{.2cm}\\
    $^1$Department of Statistics, Iowa State University \\
    $^2$Jet Propulsion Laboratory, California Institute of Technology}
      \date{}
  \maketitle
} \fi

\if0\blind
{
\title{\bf A Geospatial Functional Model For OCO-2 Data with Application on Imputation and Land Fraction Estimation}
\author{}
\date{}
\maketitle
  %\bigskip
  %\bigskip
  %\bigskip
  %\begin{center}
    %{\LARGE\bf A Geospatial Functional Model For OCO-2 Data with Application on %Imputation and Land Fraction Estimation}
%\end{center}
  %\medskip
} \fi

\bigskip
\begin{abstract} 
% 200 words
Data from NASA's Orbiting Carbon Observatory-2 (OCO-2) satellite is essential to many carbon management strategies. A retrieval algorithm is used to estimate \ce{CO2} concentration using the radiance data measured by OCO-2. However, due to factors such as cloud cover and cosmic rays, the spatial coverage of the retrieval algorithm is limited in some areas of critical importance for carbon cycle science.  Mixed land/water pixels along the coastline are also not used in the retrieval processing due to the lack of valid ancillary variables including land fraction. We propose an approach to model spatial spectral data to solve these two problems by radiance imputation and land fraction estimation. The spectral observations are modeled as spatially indexed functional data with footprint-specific parameters and are reduced to much lower dimensions by functional principal component analysis. The principal component scores are modeled as random fields to account for the spatial dependence, and the missing spectral observations are imputed by kriging the principal component scores. The proposed method is shown to impute spectral radiance with high accuracy for observations over the Pacific Ocean. An unmixing approach based on this model provides much more accurate land fraction estimates in our validation study along Greece coastlines.
%The principal component scores are modeled as heterogeneous random fields to account for the spatial dependence and footprint variability, and the missing spectral observation are imputed by kriging the principal component scores. 
\end{abstract}

\noindent%
{\it Keywords:} Functional Principal Component Analysis, Ordinary Kriging, Remote Sensing Data, Spectral Unmixing.
\vfill

\newpage
\spacingset{1.75} % DON'T change the spacing!
\section{Introduction}
\label{sec:intro}
Satellite remote sensing data has been used to provide information on many processes in the Earth system for a long time. Geophysical quantities of interest are often inferred from the radiance spectra directly observed by remote sensing instruments. A growing constellation of satellites are providing estimates of greenhouse gas concentrations globally at fine spatial resolution. One of the more recent effort to estimate the atmospheric carbon dioxide (\ce{CO2}) concentration is through the NASA's Orbiting Carbon Observatory-2 (OCO-2), which provides information on the carbon cycle at the global and regional scales \citep{elderingsci}. Several data-processing and inference stages are executed in translating the observed satellite radiance, termed Level 1 data products, into inferences on carbon sources and sinks \citep{cressiejasa}. The {\it{retrieval algorithm}} implements the estimation of \ce{CO2} concentration from Level 1 data \citep{odell2018}. For OCO-2, the primary retrieval output, or Level 2 product, of interest is $X_{CO2}$, which is the average concentration of carbon dioxide in a column of dry air extending from Earth's surface to the top of the atmosphere. The OCO-2 instrument observes high-resolution spectra of reflected sunlight at wavelengths (colors) with three spectrometers, each focusing on a narrow spectral band of the near infrared portion of the electromagnetic spectrum. The \ce{O2} A-band covers wavelengths with substantial absorption of oxygen; the weak \ce{CO2} and strong \ce{CO2} bands include spectral ranges with carbon dioxide absorption. However, OCO-2 is also sensitive to other atmospheric and surface properties, including cloud coverage and land-ocean transitions. These challenges result in a significant amount of locations with unusable data for the retrieval. Retrieval spatial coverage could improve if we are able to impute spectral observations for missing locations. For a complicated and massive data product like OCO-2, additional ancillary data is often needed and is subject to error. This includes the land fraction estimate used as an input into the retrieval. Based on the imputation algorithm proposed in this work, an unmixing approach has the potential to obtain much more accurate land fraction estimation only using measured radiance and geolocation information.

The model we developed for spectral observation imputation is based on functional data analysis, which went through great developments in both theory and methodology in recent decades. It has been successfully applied to data with various structures such as longitudinal data \citep{zhao2004functional, wang2016functional}, image data \citep{li2019spatial, wang2019simultaneous} and spatial data \citep{delicado2010, kokoszka2019some}. Functional principal component analysis (FPCA, \cite{grenander1950stochastic}), is a widely used tool in FDA. Theory on estimation and inference for FPCA vary depending on how densely the function of interest is observed, i.e., a sparse scenario \citep{yao2005} as opposed to a dense scenario \citep{ramsay}. OCO-2 represents an instance of dense functional data, with each observation a spectral observed at a specific location and time and treated as a function of wavelength. Each spectral consists of thousands of radiances measurements at different wavelengths. Modeling OCO-2 data as dense functional observations, we are able to conduct FPCA to reduce the dimension of the problem and reconstruct radiance function from the reduced space for the task of imputation and land fraction estimation. Several specific challenges outlined in Section \ref{data} requires further consideration in modeling. One of them is that at any given time, the cross-track measurements on OCO-2 are aggregated into eight discrete {\it{footprints}} to meet the storage and downlink bandwidth limits. The footprints correspond to different physical locations on the instrument and are characterized separately in OCO-2 data processing. As a result, the radiance measurements associated with different footprints hold different data characteristics such as mean and variation, which has to be taken into account in modeling.

Another challenge is that by the nature of the physical atmospheric process, the functional observations from OCO-2 in different locations and time are dependent, and the dependence is extreme with locations nearby. The study of spatially dependent functional data has gained much attention recently \citep{martinez2020recent}. Much of the work is built for functional data indexed by locations on grids or lattice \citep{zhang2016, kuenzer2020principal}. While some works are motivated by functional data on irregular spatial locations, they mostly focus on different purposes. \cite{gromenko2012estimation} developed a test for the correlation of two functional spatial fields, \cite{menafoglio2016kriging} proposed a unifying framework for kriging in functional random fields. Although \cite{ruggieri2018comparing} aimed at imputing large gaps in space, it models data as a smooth function across space and time, which is not reasonable for our case. The model in \cite{liu2017} is the most similar to our approach. They treat spatial-temporal data as spatially indexed functional data in time and developed asymptotic results as well as two tests for separability and isotropy, assuming the functional data is observed sparsely in time. In our OCO-2 application, the hyperspectral remote sensing data are observed in both space and time. More importantly, modeling different characteristics from each footprint together is crucial to our analysis, which does not fit into existing frameworks. We model the principal component scores as spatially dependent processes and the mean function as a linear structure dependent on both locations and footprints. This model differs from both \cite{liu2017} and \cite{ruggieri2018comparing} in that we need to have a non-trivial model for the variation among footprints, and we have densely observed functional data and does not need to make any smoothness assumption in wavelength direction.

Treating radiance observations as dense functional data, we propose a geospatial functional model for spatial spectral data to impute missing hyperspectral radiance in the OCO-2 data. Spatial dependence among radiance function is introduced by modeling the FPC scores as spatial processes. Footprint-specific mean radiance functions and measurement error process are specified to account for the heterogeneity across footprints observed in the OCO-2 data. Unlike most dense functional data methodology, this point-wise approach does not need smoothness assumptions and fits spectral data naturally. We develop and implement practical algorithms for parameter estimation and imputation, and establish asymptotic consistency and convergence rate results for the procedure. Simulation studies and validation studies using OCO-2 data have shown that our algorithm achieved high accuracy for radiance imputation at footprints over water. An unmixing approach based on the imputation is also developed to estimate the land fraction over mixed footprints and shown in real OCO-2 data to provide much more accurate land fraction than those offered in the OCO-2 data.

Besides the statistical methodology contributions, our methodology can significantly improve the OCO-2 data by increasing the spatial coverage of the \ce{CO2} retrieval algorithm. 
%With successful radiance imputation, Level 2 retrievals could be able to cover those locations with missing observations. Also, an accurate estimate of land fraction from the radiance data could facilitate retrievals in mixed land/ocean area. 
The imputation approach allows for additional successful Level 2 retrievals, where radiances can be reliably imputed. The OCO-2 operational Level 2 retrieval algorithm has different configurations for land and ocean soundings. Retrievals are not attempted for mixed land/ocean soundings \citep{odell2018} partly due to a lack of reliable land fraction estimates. The more consistent and accurate estimate of land fraction from our model can facilitate retrievals in these mixed cases, further expand the spatial coverage of the OCO2 level 2 retrieval. The unmixed land fraction estimates could also supplement OCO-2's geolocation information.

The rest of the paper is organized as follows. We introduce the structure of OCO-2 data and the variables considered in Section \ref{data}. Then we propose a geospatial functional model for spatial spectral data with different characteristics on different footprints in Section \ref{model}. In Section \ref{estimation}, we discuss the estimation and prediction based on our data model for the purpose of radiance imputation in water area and land fraction correction in mixed regions, and derive asymptotic results that justify the procedure. Lastly, in Section \ref{application}, the spatial coverage and accuracy of the retrieval algorithm are shown to be improved when the proposed radiance imputation algorithm and land fraction estimation procedure were applied to OCO-2 Level 1 data. We conclude in Section \ref{discussion} with some discussion on possible future work. Technical proofs and additional numerical
results are relegated to the Supplementary Material.

\section{OCO-2 Data} \label{data}
OCO-2 is part of a constellation of polar-orbiting satellites known as the A-train and completes approximately 15 orbits per day. The satellite crosses the equator in the early afternoon local time on each orbit. 
\begin{figure}[ht]
    \centering
    \includegraphics[trim = 6cm 0cm 6cm 0cm, clip = true, width = 0.9\textwidth]{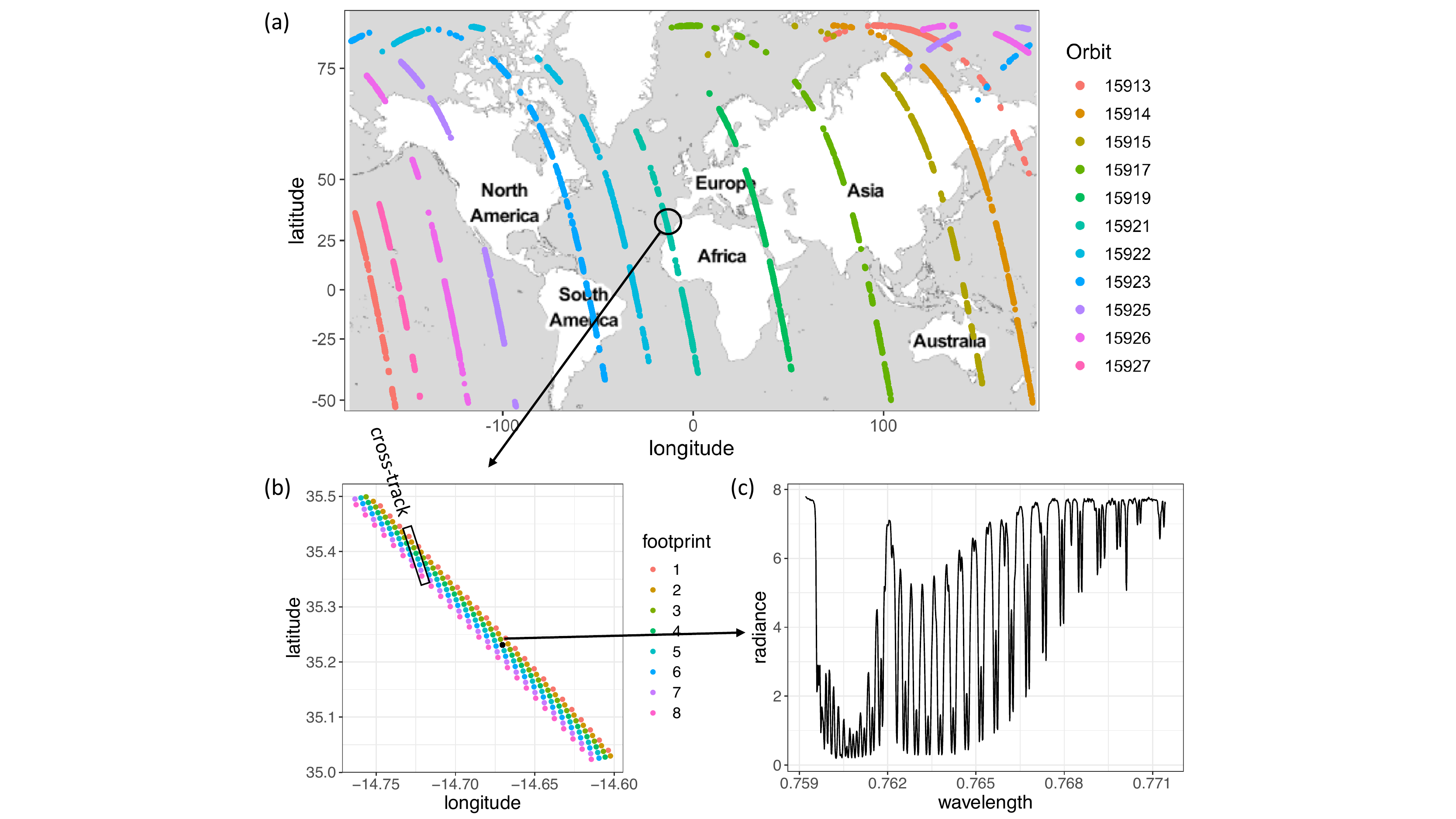}
    \caption{OCO-2 data illustration: (a) orbits completed on 2017-06-29 under glint mode; (b) spatial layout of 8 cross-track footprints from a partial region of orbit 15921; (c) radiance function of \ce{O2} band at (35.2307, -14.6704) on orbit 15921.}
    \label{fig:data_def}
\end{figure}
Orbits alternate between nadir and glint observing modes. In nadir mode, data are collected directly below the satellite, minimizing the optical path length through the atmosphere. In the glint mode, the instrument points at an angle directed toward the glint spot, allowing a high signal over the ocean \citep{ElderingEtal2017}. Our analysis focuses on glint observations, which are available over land and ocean, including along coastlines. As shown in the top panel (a) of Fig \ref{fig:data_def}, glint orbits completed in one day can be far apart and will cover worldwide after 16 days. The OCO-2 field of view is approximately 10 km wide along an orbit track, and this spatial orientation translates to physical positions on the focal plane arrays (FPAs) for each of the spectrometers on the instrument. In order to meet bandwidth limitations for storage and downlink from the satellite, the cross-track spectra are aggregated into eight discrete {\it{footprints}}. The panel (b) of Fig \ref{fig:data_def} is a zoom-in presentation of a partial region of orbit 15921, displays how eight footprints are distributed along cross-tracks. Because the footprints correspond to different physical locations on the instrument, they are characterized separately in OCO-2 data processing \citep{CrispCal}.

\begin{figure}[ht]
   \centering
    \begin{subfigure}[t]{.49\textwidth}
        \centering
 		\includegraphics[trim = 10cm 0 10cm 5cm, clip = true, width = \textwidth]{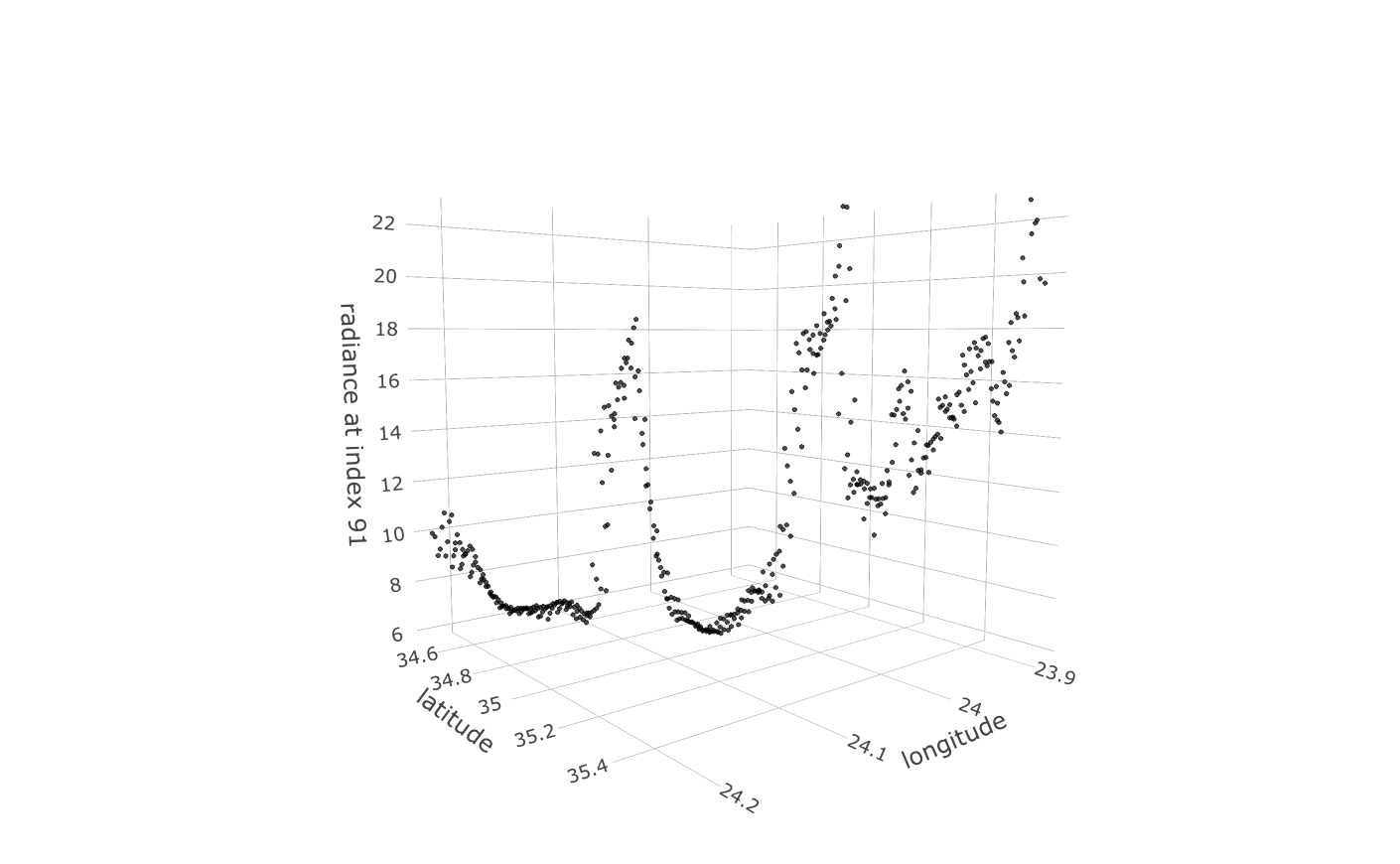}
 		\caption{}
 		\label{fig:91_rads_3Dplot}
    \end{subfigure}
 	\begin{subfigure}[t]{0.49\textwidth}
 		\centering
 		\includegraphics[width = \textwidth]{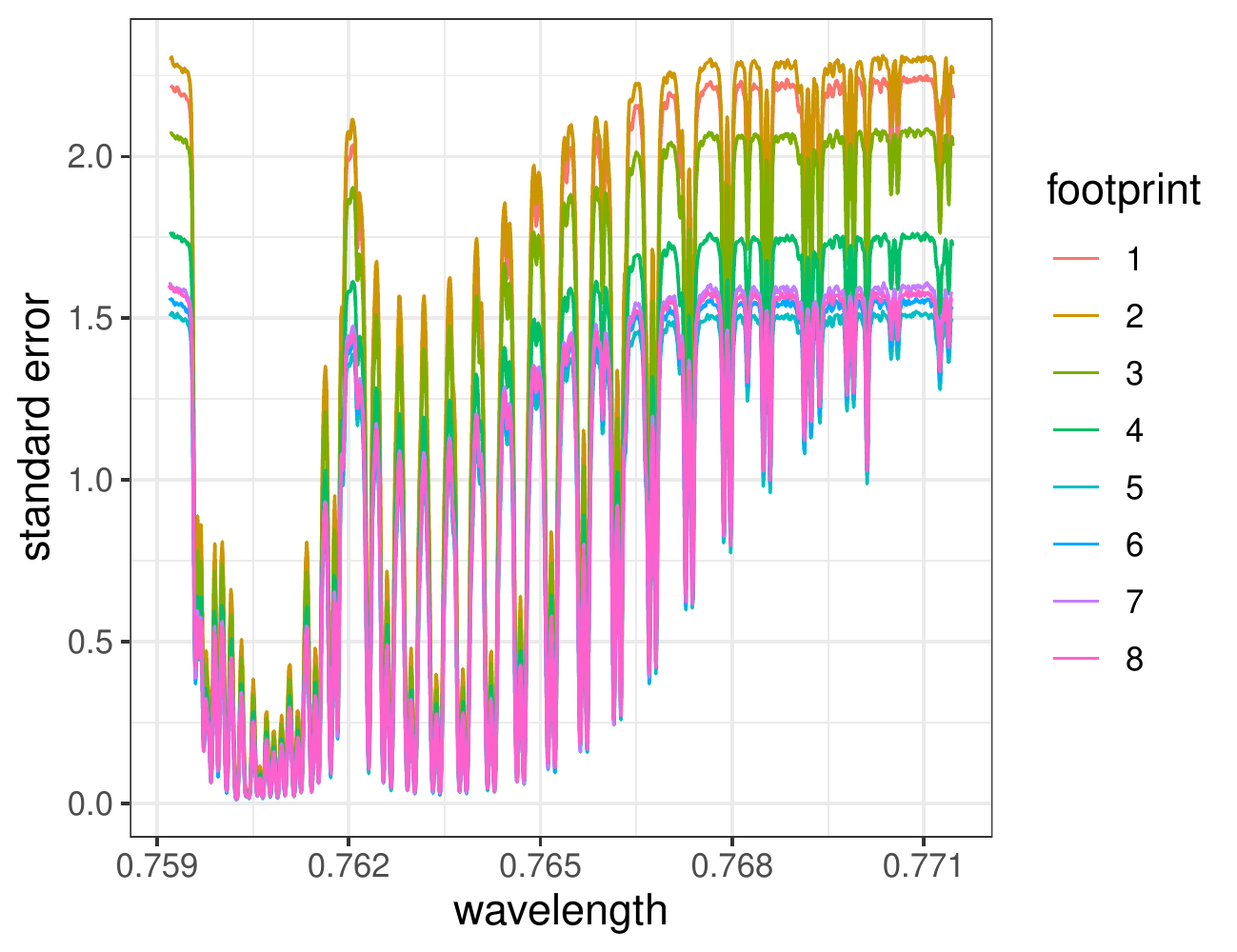}
 		\caption{}
 		\label{fig:measerror_ftprint}
 	\end{subfigure}
 	%\begin{subfigure}{.49\textwidth}
 	%	\centering
 	%	\includegraphics[width = \textwidth]{./sd_vs_rad}
 	%\end{subfigure}
 	\caption{Measurement error at different wavelengths depends on footprints: (a) radiance at wavelength index 91 in a partial area of orbit 10575; (b) measurement error estimates by footprint in region within latitude [34.3, 34.8] on orbit 10575.}
 \end{figure}
We use OCO-2 Level 1 products, which include latitude, longitude, orbit, footprint, and land fractions for locations and times of interest. Each unique location (corresponding to one footprint) and time defines a single observation, or {\it{sounding}}. Level 1 data also includes the wavelengths and corresponding measured radiances for each sounding. As discussed in the introduction, the wavelengths and radiances are categorized into 3 bands, with 1016 observations made for each before filtering. For simplicity here, we will only consider radiances in the \ce{O2} band, which refers to wavelengths (colors) within the 0.765 micron molecular oxygen A-band (Fig \ref{fig:data_def}c). The wavelength indices, $j = 1, \ldots, 1016$ correspond to physical positions on the OCO-2 spectrometer, and each position is termed as {\it{spectral sample}} in the OCO-2 documentation. The wavelength corresponding to each sample varies slightly across soundings in space and time, but the within-sample variability is small compared to the monotonic change across samples with increasing $j$. Therefore, the recorded wavelength can be regarded as fixed, and it is equivalent to treat radiance as a function in discrete wavelength indices. Finally, throughout this paper, the measured radiance is reported in the unit of $10^{19} (\text{photons} \; \text{m}^{-2} \text{sr}^{-1} \mu \text{m}^{-1})$. 
 
  \begin{figure}[ht]
 	\begin{subfigure}{.49\textwidth}
 		\centering
 		\includegraphics[width = \textwidth]{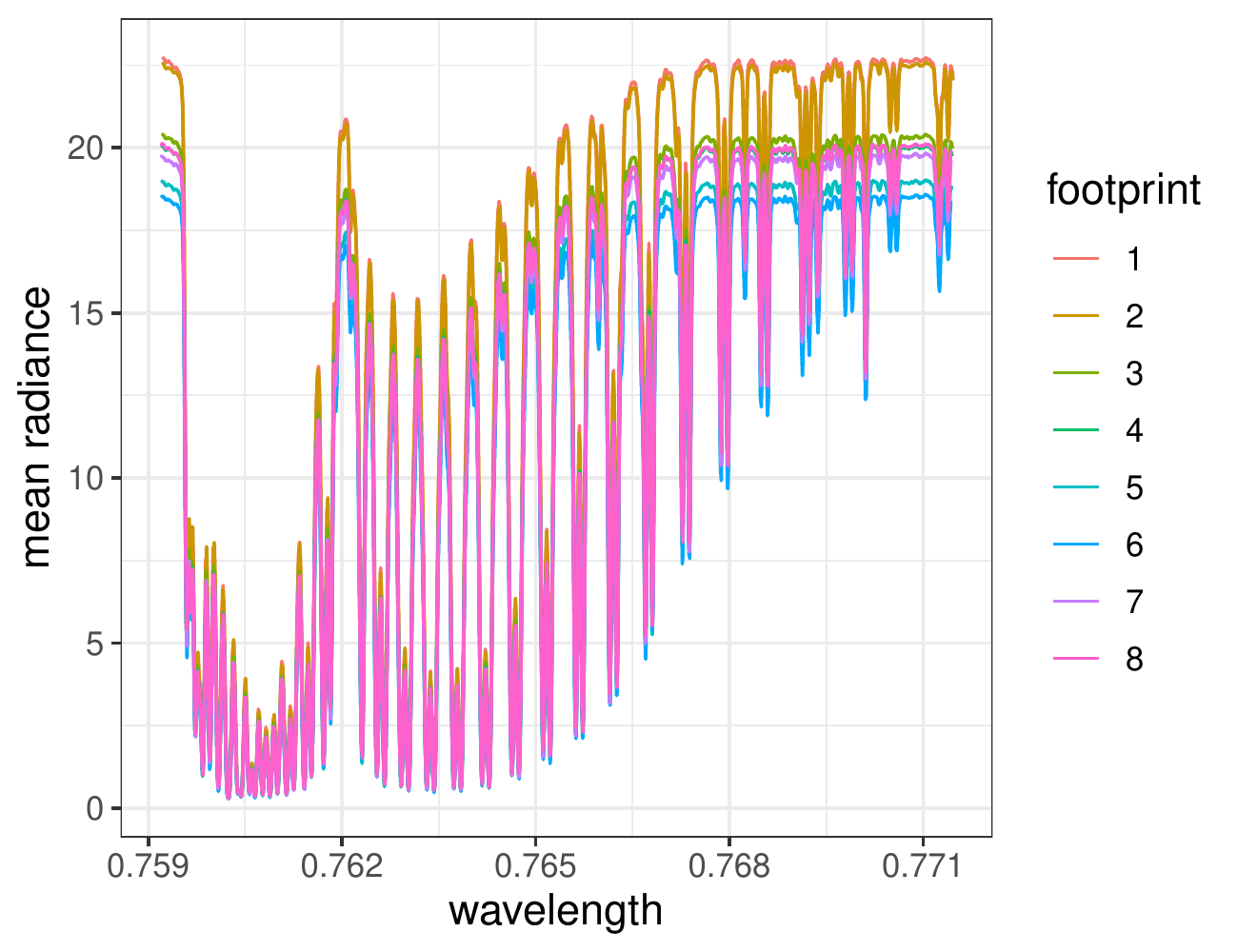}
 	\end{subfigure}
 	\begin{subfigure}{.49\textwidth}
 		\centering
 		\includegraphics[width = \textwidth]{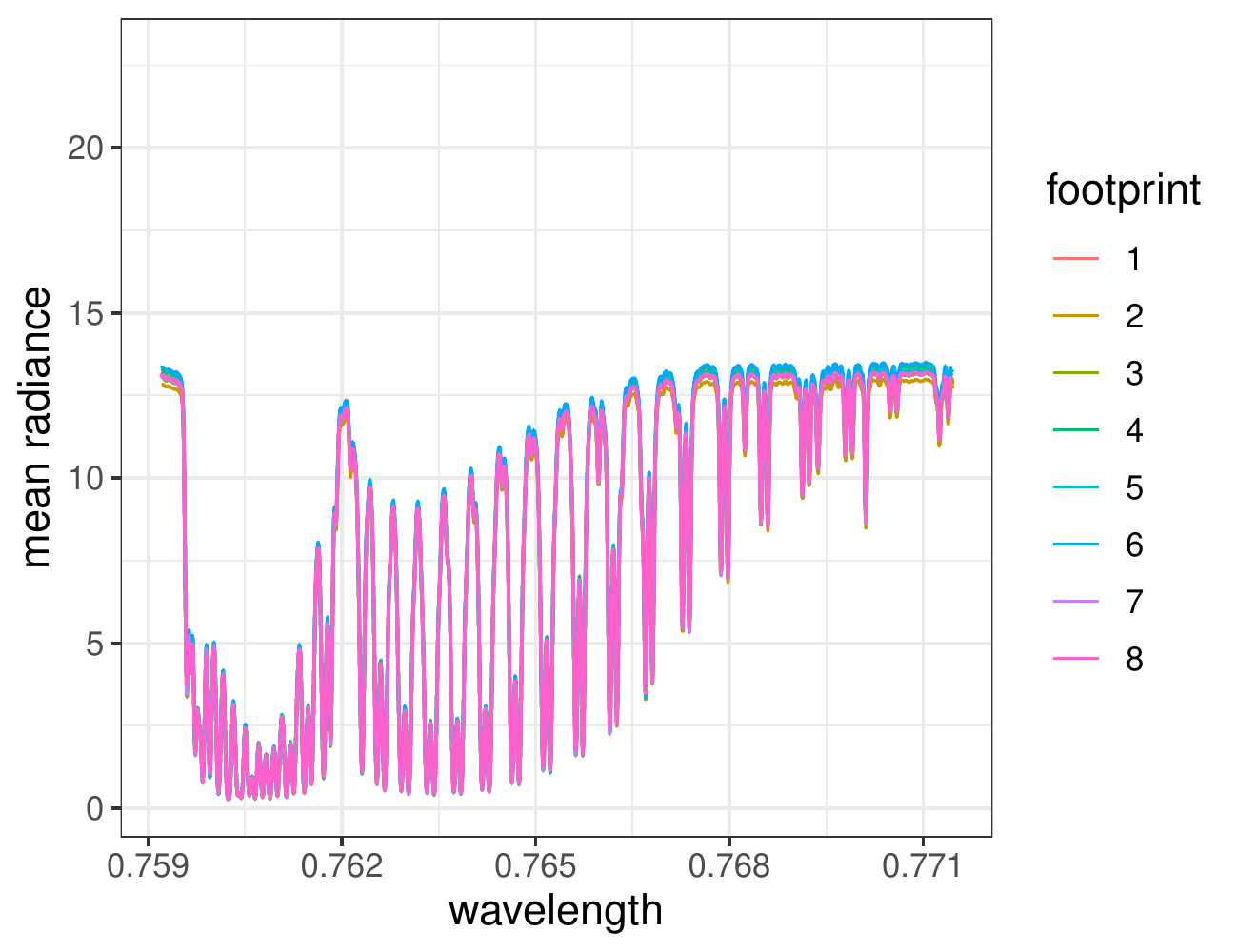}
 	\end{subfigure}
     \caption{Mean radiance by footprint in regions within latitude [34, 34.5] and [35, 35.5] on orbit 10575}
 \label{fig:mean_ftprint}
 \end{figure}
Due to the fact that eight footprints effectively constitute eight instruments, we found a significant difference in the variability of radiance across footprints. Preliminary analysis revealed that the standard error of measurement error varies at different wavelengths, which also depends on footprints (Fig \ref{fig:measerror_ftprint}). Here standard deviation is estimated by 2nd order differencing along spatial direction (eqn. \ref{meas_error_est}) since radiance as a function of location is smooth in that direction as shown in Fig \ref{fig:91_rads_3Dplot} compared to Fig \ref{fig:data_def}c. Due to this variation, the standard approaches to handle measurement errors in the dense FDA \citep{castro1986, yao2003} are no longer applicable and need to be adapted to the specific case concerned in this paper. Besides that, as demonstrated in Fig \ref{fig:mean_ftprint}, the naive moving-window mean radiance estimate (averaging observations within a window of locations) varies across locations and footprints. These data characteristics motivate motivated us to specify footprint-dependent mean functions and error process in our models detailed in the following two sections.

\section{Geospatial Functional Model based on FPCA}\label{model}

\subsection{Functional Data}
For the \ce{O2} band we are interested in, we treat the radiance $f(w; \bds_{i})$ in a sounding location $\bds_{i}\trans = (L_{i}, l_{i}) = (latitude, longitude)$ as a function of wavelengths, i.e., function in $w$ on a compact domain $\mathcal{W}$. 
%Hence $f(\cdot; \bds_i)$ is defined as a function on a discrete domain $\mathcal{W}$ with finitely many points, i.e., $\mathcal{W} = \{w \in \mathbb{Z}: 1 \leq w \leq 1016\}$. 
Here the index $i$ indicates both the footprint and the time and is ordered by footprint (from 1 to 8) followed by the sounding time. $\bds_{i}$ is the spatial location of the radiance for the particular footprint/sounding time combination indicated by $i$, and $\mathcal{S} = \{\bds_1, \ldots, \bds_N\}$ is the area of interest. We model measured radiance in wavelength at location $\bds_i$ as
\begin{align}\label{fpca_model0}
r(w; \bds_{i}) = f(w; \bds_{i}) + \epsilon_{i}(w), \quad  i = 1, \ldots, N, 
\end{align}
where we assume the measurement error process $\epsilon_{i}(w)$ is independent of $f(w; \bds_i)$. Within an orbit, the footprint of $\bds_i$ is denoted as $p_i = q(\bds_i)$, where $q$ is a known function which determines the footprint corresponding to a location through the status of the satellite. Inspired by the findings in Fig \ref{fig:measerror_ftprint}, we assume that the measurement error $\epsilon_{i}(w)$ is distributed as $\sigma_{p_i}(w)e_{i}(w)$, where $e_{i}(w), i = 1, \ldots, N$ are i.i.d. stochastic processes with mean zero and variance 1, and $\sigma_p(w)$ is the variance function for footprint $p$, a bounded function on $\mathcal{W}$ for $p = 1, \ldots, 8$. Because of the real data characteristic shown in Fig \ref{fig:mean_ftprint}, we model the mean function $\mu(w; \bds_i) = \E\{f(w; \bds_i)\}$ to be dependent on location $\bds_i \in \mathcal{S}$ and continuous in $w \in \mathcal{W}$. Furthermore, we observe that measured radiance varies across different footprints and can be described by a linear relationship with location as covariates, and specify the mean function as a linear model with footprint specific coefficients for each fixed wavelength,
\begin{align} \label{mean_model}
 \mu(w; \bds_i) = \sum_{p=1}^8 I\{q(\bds_i) = p\} \{\beta_0^{p} (w) + \bds_i \trans \bdbeta_1^{p} (w)\}, \quad w \in \mathcal{W}.
\end{align}

Let $R_f(w, w')$ be the continuous covariance function of $f(w; \bds_i)$ at any given $\bds_i \in \mathcal{S}$. By the Karhunen–Lo\`{e}ve theorem,
\begin{align} \label{fpca_model0.5}
    f(w; \bds_i) - \mu(w; \bds_i) = \sum_{k = 1}^{\infty} \xi_k(\bds_i) \phi_k(w), \quad w \in \mathcal{W}, 
\end{align}
where $\xi_k(\bds_i)$ are uncorrelated principal component scores with mean zero and variance $\lambda_k$, and $\phi_k$ are eigenfunctions of $R_f$. 
By decomposing $\epsilon_i(w)$ onto the orthonormal basis $\phi_k$ in the Hilbert space, model (\ref{fpca_model0}) can be written as 
\begin{align} \label{fpca_model0_alt}
    r(w; \bds_i) = \mu(w; \bds_i) + \sum_{k = 1}^{\infty} \{\xi_k(\bds_i) + e_{ik}\} \phi_k(w), \quad w \in \mathcal{W}, 
\end{align}
where $e_{ik} = \int \epsilon_i(w)\phi_k(w) dw$ are independent random variables with mean zero and variance $\tau_k(p_i)$ depending on the location's footprint. It is noted that $e_{ik}$ are not necessarily independent over $k \geq 1$. %Also we can see connections between two expressions,  
Here $\sum_{k=1}^{\infty} e_{i,k}\phi_k(w)$ in (\ref{fpca_model0_alt}) represents the error process $\epsilon_i(w)$ in (\ref{fpca_model0}), and each footprint has a specific covariance function for its error process, 
\begin{align*}
R_{\epsilon, p}(w, w') & \equiv \cov\{\epsilon_i(w), \epsilon_i(w')\} \\
& = \sigma_{p}(w) \sigma_{p}(w') R_e(w, w') = \sum_{k=1}^{\infty} \sum_{k'=1}^{\infty} \cov(e_{ik}, e_{ik'})\phi_k(w) \phi_{k'}(w')
\end{align*}
for $\bds_i \in \mathcal{S}_p = \{\bds_i: q(\bds_i) = p\}$, where $R_e(w, w') = \cov\{e_i(w), e_i(w')\}$, $w,w' \in \mathcal{W}$ is assumed to be continuous.

\subsection{Spatial Dependence}
We model the functional spatial dependence among $f(w; \bds_i)$ through the FPC scores $\xi_k(\bds_i)$, and assume they are mean zero, second order stationary and isotropic random fields in space. Let $G_k(\cdot)$ be the covariance function for $\xi_k(\bds_i)$, and let $u_k(\bds_i) = \xi_k(\bds_i) + e_{ik}$. For any two points $\bds_i$ and $\bds_{i'}$, we have
\begin{align}
H_k(\bds_i, \bds_{i'}) \equiv \cov\{u_k(\bds_i), u_k(\bds_{i'})\} = G_k\{d(s_i, s_{i'})\} + \tau_k(p_i)I(\bds_i = \bds_{i'}),
\end{align}
where $d(\cdot, \cdot)$ denotes the great circle distance between two geolocations, and $H_k(\bds_i, \bds_i) = \lambda_k + \tau_k(p_i)$ is implied.
%zhu what is $\lambda_k$?
Based on model representation (\ref{fpca_model0_alt}), $u_k(\bds_i)$ is the projection of the combined process $f(w; \bds_i) + \epsilon_i(w)$ onto the Hilbert space spanned by orthonormal basis $\phi_k$. Because of the error component $e_{ik}$, the covariance function $H_k(\cdot)$ has a footprint-specific nugget effect. Note that we assume $G_k(\cdot)$ does not differ among footprints, which is supported by the empirical evidence. Models assuming the footprint-specific $G_k^p(\cdot)$ did not improve the prediction performance. 

In practice, the first few principal components capture most of the variation. Assuming $K < \infty$, we consider the truncated version,
\begin{align} \label{fpca_model1}
f(w; \bds_i) = \mu(w; \bds_i) + \sum_{k = 1}^{K} \xi_k(\bds_i) \phi_k(w), \quad w \in \mathcal{W}.
\end{align}
Thus, $f(w; \bds_i)$ is assumed to be a realization of a stationary spectral process in our study region $\mathcal{S}$ (approximately $\pm 0.5^{\circ}$ in latitude). It is assumed to have non-homogeneous mean radiance function $\mu(w, \bds_i)$ across spatial locations $\bds_i \in \mathcal{S}$, and the residuals defined as $f(w; \bds_i) - \mu(w; \bds_i)$ is approximated by the projection of the error process to the function space spanned by the first $K$ eigenfunctions.

\subsection{Mixing Process of Water and Land}
\label{model:mixing}
Let $\mathcal{M}$ be an area with mixed pixels such as a coastline crossing, where land fractions are between 0 and 1. In general, physical characteristic including radiance and $X_{CO2}$ concentration of mixed locations could be influenced by both water and land areas. It is natural to model the observed radiance process $r(w; \bds_i), w \in \mathcal{W}$ at mixed sounding location $\bds_i \in \mathcal{M}$ as a linear combination of possible water radiance $f_w(w; \bds_i)$ and land radiance $f_l(w; \bds_i)$ as follows, 
\begin{align} \label{mix_model}
r(w; \bds_i) = \alpha_i f_l(w; \bds_i) + (1-\alpha_i) f_w(w; \bds_i) + \epsilon_{i}(w),
\end{align}
where $\alpha_i$ is the land fraction of the mixed location $\bds_i$. 
%Also, the measurement error process $\epsilon_{i}(w)$ are independent with water and land radiance, and identically, independently distributed as $\tau(w) e_i(w)$ with $\tau(w)$ as a bounded function on $\mathcal{W}$. 
The components of the mixture, the two underlying processes $f_l(w; \bds_i)$ and $f_w(w; \bds_i)$, are modeled as we described in (\ref{fpca_model1}), i.e.,
\begin{align*}
f_l(w; \bds_i) & = \mu_l(w; \bds_i) + \sum_{k = 1}^{K_l} \xi_k^l(\bds_i) \phi_k^l(w), w \in \mathcal{W}, \\
f_w(w; \bds_i) & = \mu_w(w; \bds_i) + \sum_{k = 1}^{K_w} \xi_k^w(\bds_i) \phi_k^w(w), w \in \mathcal{W}.
\end{align*}

\section{Estimation and Prediction} \label{estimation}

\subsection{Dense FPCA}
Without loss of generality, we assume the trajectories of $r(w; \bds_i)$ are fully observed on the domain set $\mathcal{W}$. Let $\bdx_p$ be the matrix of location covariates in the footprint-specific set $\mathcal{S}_p = \{\bds_i: q(\bds_i) = p\}$, and hence $\{\bds_1, \ldots, \bds_{N}\}\trans = (\bdx_1\trans, \cdots, \bdx_8\trans)\trans$. Similarly, we denote the measured radiance at wavelength $w \in \mathcal{W}$ as $\BY(w) = (\BY_{1}(w)\trans, \cdots, \BY_{8}(w)\trans )\trans$, where $\BY_{p}(w)$ is the vector of measured radiance $r(w; \bds_i)$ with $\bds_i \in \mathcal{S}_p$. So the full design matrix $\BX$ corresponding to the mean function model (\ref{mean_model}) is
\begin{align*}
    \BX = \begin{bmatrix}
   \bdone & \bdx_1 & & & & \\
   & & \bdone & \bdx_2 &  & \\
   & & & \ddots & \ddots & \\
   & & & & \bdone & \bdx_8 
\end{bmatrix}, 
\end{align*}
where the empty entries are all zeros. As a linear model for each wavelength, the footprint-specific coefficients $\bdbeta(w) = \{\beta_0^1(w), \bdbeta_1^1(w)\trans, \ldots, \beta_0^8(w), \bdbeta_1^8(w)\trans \}\trans $ at wavelength $w \in \mathcal{W}$ can be estimated by
\begin{align} \label{beta_estimate}
\wh \bdbeta(w) = (\BX\trans \BX)^{-1} \BX\trans \BY(w),
\end{align}
further we can estimate the location-dependent mean function as 
\begin{align} \label{mean_estimate}
    \wh \mu(w; \bds_i) = \sum_{p=1}^8 I\{q(\bds_i) = p\} \{\wh \beta_0^{p} (w) + \bds_i \trans \wh \bdbeta_1^{p} (w)\}, \quad w \in \mathcal{W},
\end{align}
where $I(\cdot)$ is an indicator function. 

Based on the data model (\ref{fpca_model0}), there are two parts in the covariance of observed radiance, 
\begin{align*}
\cov\{r(w; \bds_i), r(w'; \bds_i)\} = R_f(w, w') + \sum_{p=1}^8 I \{q(\bds_i) = p\} R_{\epsilon, p}(w, w'), \quad w, w \in \mathcal{W}, 
\end{align*}
for any $\bds_i \in \mathcal{S}$. Define $\Delta_p(w; \bds_i) = r(w; \bds_{i_1}) - 2r(w; \bds_i) + r(w; \bds_{i_2})$, where $\bds_{i_1} \in \mathcal{S}_{p}$ is the nearest location of $\bds_i$ in the positive latitude direction and $\bds_{i_2} \in \mathcal{S}_p$ is the nearest locations of $\bds_i$ in the negative latitude direction.
%zhu Is this correct?
We can estimate the covariance function for footprint $p$ by averaging the second order differencing along the spatial direction, 
\begin{align} \label{meas_error_est}
\wh R_{\epsilon, p}(w, w')  & = \dfrac{1}{6 \wt N_p} \sum_{\substack{\bds_i \in \mathcal{S}_p}} \Delta_p(w; \bds_i)\Delta_p(w'; \bds_i),
\end{align}
where $\wt N_p$ is the number of points having both valid $\bds_{i_1}$ and $\bds_{i_2}$ in $\mathcal{S}_p$.
%In practice, considering cross covariance is not significant in real data, we use the simplified version instead, 
%\[
%\wh R_{\epsilon, p}(w, w') = I (w=w') \dfrac{1}{6N_p} \sum_{\substack{\bds_i \in \mathcal{S}_p}} \Delta_p^2(w; \bds_i).
%\]
Next we can estimate the covariance function as 
\begin{align} \nonumber
 \wh R_f(w, w') = & \dfrac{1}{N-1}\sum_{\bds_i \in \mathcal{S}} \left\{r(w; \bds_i) - \wh \mu(w; \bds_i)\right\} \left\{r(w'; \bds_i) - \wh \mu(w'; \bds_i)\right\} \\ \label{cov_func_est}
    & - \dfrac{1}{N-1} \sum_{p = 1}^{8}N_p \wh R_{\epsilon, p}(w, w').
\end{align}
As shown in Fig \ref{fig:data_def}, locations within each footprint are aligned in one line, thus the nearest two locations can be referred through index after ordering locations. In practice, we enumerate the locations ordered by latitude as $\bds_i, i = 1, \ldots, N_p$ and set $\bds_{i_1} = \bds_{i-1}, \bds_{i_2} = \bds_{i+1}$.

Then eigenfunctions are estimated by the sample eigenfunction of $\wh R_f$, and in practice is obtained by matrix decomposition with respect to wavelength,
\begin{align} \label{cov_decompose}
\wh R_f(w, w') = \sum_{k=1}^K \wh \lambda_k \wh \phi_k(w) \wh \phi_k(w').
\end{align}
The number of principal components $K$ is selected as the truncation to which point certain variance can be explained, i.e., fraction of variance explained (FVE), 
\begin{align*}
    \mathrm{FVE}(K)=\frac{\sum_{k=1}^{K} \lambda_{k}}{\sum_{j=1}^{\infty} \lambda_{j}}, \quad \widehat{ \mathrm{ FVE}}(K)=\frac{\sum_{k=1}^{K} \wh \lambda_{k}}{\sum_{j=1}^{\infty} \wh \lambda_{j}}.
\end{align*}
The corresponding $u_k(\bds_i)$ are estimated by numerical integration,
\begin{align}  \label{xik_estimate}
\wh u_k(\bds_i) = \int \left\{r(w; \bds_i) - \wh \mu(w; \bds_i)\right\}  \wh \phi_k(w) dw,
%& = \sum_{w \in \mathcal{W}}  \left\{r(w; \bds_i) - \wh \mu(w; \bds_i)\right\}  \wh \phi_k(w) 
\end{align}
and the variance of $e_{ik}$ is estimated similarly,
\begin{align} \label{tauk_estimate}
 \wh \tau_k(p) = \int \int \wh R_{\epsilon, p}(w, w') \wh \phi_k(w) \wh \phi_k(w')dw dw',
\end{align}
for $p = 1, \ldots, 8$. 

\subsection{BLUP for Principal Component Scores}
Suppose $\bds_0 = (L_0, l_0)$ is a sounding location with no measured radiance observed. Based on our assumptions, the principal components scores $\xi_k(\bds_i)$ form a second order stationary and isotropic random field. We use the best linear unbiased predictor (BLUP) to predict $\xi_k(\bds_0)$ using estimated $u_k(\bds_i)$. Let $\gamma_k(h; \theta_k)$ be the semivariogram of $\xi_k(\bds_i)$, which is a function of distance $h$ and with parameters $\theta_k \in \Theta$. We define the  semivariogram estimator $\wh \gamma_k(h)$ as 
\begin{align} \nonumber
    \wh \gamma_k(h_l) = & \dfrac{1}{2N(h_l)} \sum_{d(\bds_i,\bds_j)=h_l} \{\wh u_k(\bds_i) - \wh u_k(\bds_j)\}^2 \\ \label{sample_variogram}
     & - \dfrac{1}{2N(h_l)} \sum_{d(\bds_i,\bds_j)=h_l} \{\wh \tau_k(p_i) + \wh \tau_k(p_j)\}, \quad l = 1, \ldots, L,
\end{align}
for a collection of $L$ distance bins, where $N(h_l)$ is the number of pairs with distance $h_l$.
We use weighted least squares to estimate $\theta_k$ with $V_k(\theta_k)$ as the weight matrix for the $k$th component, 
\begin{align} \label{thetak_estimate}
    \wh \theta_k = \underset{\theta_k \in \Theta}{\arg\min} \{\wh \bdgamma_k - \bdgamma_k(\theta_k)\}\trans V_k(\theta_k)  \{\wh \bdgamma_k - \bdgamma_k(\theta_k)\},
\end{align}
where $\wh \bdgamma_k = \{\wh \gamma_k(h_1), \ldots, \wh \gamma_k(h_L)\}\trans$ and $\bdgamma_k(\theta_k) = \{\gamma_k(h_1; \theta_k), \ldots, \gamma_k(h_L; \theta_k)\}\trans$. In practice, the exponential model is used as our choice of $\gamma_k(h; \theta_k)$. For $V_k(\theta_k)$, diagonal matrix using $N(h_l)/h_l^2$ and $N(h_l)$ are both valid.

Let $\bdu_k = \{u_k(\bds_1), \ldots, u_k(\bds_N)\}\trans$, $\BSigma_k = \var(\bdu_k)$ and $\bdnu_k = \cov\{\xi_k(\bds_0), \bdu_k\}$, which can be constructed by the covariance function $G_k(\cdot)$ and $\tau_k(p), p = 1, \ldots, 8$. With $\wh \theta_k$ and $\wh \tau_k(p)$ computed using the methods described above, it is straight forward to obtain $\wh \BSigma_k$ and $\wh \bdnu_k$. In our implementation, the estimated weighting scores $\wh u_k(\bds_i)$ is used in place of $u_k(\bds_i)$ in calculation, and we estimate $\xi_k(\bds_0)$ by the plug-in ordinary kriging predictor, 
\begin{align} \label{xik_predict}
\wh \xi_k(\bds_0) = \wh \kappa + \wh \bdnu_k\trans \wh \BSigma_k^{-1} (\wh \bdu_k - \bdone \wh \kappa),
\end{align}
where $\wh \kappa = \left(\bdone \trans \wh \BSigma_k^{-1} \bdone \right)^{-1} \bdone \trans \wh \BSigma_k^{-1}\wh \bdu_k$. Finally the radiance function at the location $\bds_0$ is predicted as
\begin{align} \label{predict_eqn}
\wh f(w; \bds_0) = \wh \mu(w; \bds_0) + \sum_{k=1}^{K} \wh \xi_k(\bds_0) \wh \phi_k(w), \quad w \in \mathcal{W},
\end{align}
which is used as the imputation at location $\bds_0$.

\subsection{Land Fraction Estimation}
\subsubsection{Unmixing Approach}
For any mixed location $\bds_i$, we can find the nearest homogeneous water area $\mathcal{S}_w$ and land area $\mathcal{S}_l$. Using the imputation method described above (\ref{predict_eqn}), we can obtain two imputation of the radiance from the model on $\mathcal{S}_w$ and $\mathcal{S}_l$ respectively. On any fixed wavalength $w_j \in \mathcal{W}$, $j = 1, \ldots, m_i$,  the water and land radiance estimation are $\wh f_w(w_j; \bds_i)$ and $\wh f_l(w_j; \bds_i)$. Under the mixing model (\ref{mix_model}), we can estimate land fraction $\alpha_i$ by minimizing the sum of squared error loss $||r(w_j; \bds_i) - \alpha \wh f_l(w_j; \bds_i) - (1-\alpha) \wh f_w(w_j; \bds_i) ||_2^2$, which has the solution 
\begin{align} 
    \label{unmixing}
    \wh \alpha_i^k & = \dfrac{\sum_{j = 1}^{m_i} \{r(w_j; \bds_i) - \wh f_w(w_j; \bds_i)\} \{\wh f_l(w_j; \bds_i) - \wh f_w(w_j; \bds_i)\}}{\sum_{j = 1}^{m_i} \{\wh f_l(w_j; \bds_i) - \wh f_w(w_j; \bds_i)\}^2}.
\end{align}
Although model (\ref{mix_model}) does not assume constant variance through different wavelengths, we use ordinary least square to estimate land fractions, because it would be impossible to estimate a suitable covariance function of measurement error using limited points in mixed region. Ordinary least squares performs well in simulation and real data. 
To demonstrate the value of this kriging based unmixing approach, we compare it with a simple linear interpolation method without accounting for the measurement errors. Suppose the nearest land footprint of $\bds_i$ is $\bds_i^l$ and water footprint is $\bds_i^w$, this method estimate $\alpha$ by minimizing $||r(w_j; \bds_i) - \alpha r(w_j; \bds_i^l) - (1-\alpha) r(w_j; \bds_i^w) ||_2^2$, which has the solution 
\begin{align}
    \label{interpolation}
	\wh \alpha_i^t & = \dfrac{\sum_{j = 1}^{m} \{r(w_j; \bds_i) - r(w_j; \bds_i^w) \} \{r(w_j; \bds_i^l) - r(w_j; \bds_i^w) \}}{\sum_{j = 1}^m \{ r(w_j; \bds_i^l) - r(w_j; \bds_i^w)\}^2}.
\end{align}

\subsubsection{Simulation}
We conducted a simulation study imitating OCO-2 data to illustrate the advantage of the kriging based unmixing approach (\ref{unmixing}) over the interpolation method (\ref{interpolation}). For simplicity, data were only simulated for a single footprint, with locations equally spaced in latitude and longitude. The location in the middle (35.49881, 23.83578) is taken to be a mixed site with land fraction between 0 and 1. The portion south of the middle location is assumed to be water, and to the north is assumed to be land. The layout is similar to Fig \ref{fig:data_def}b except that we consider only one footprint and the latitude range is larger. The hyper-spectral observations at these fixed locations are simulated based on the model in Section \ref{model:mixing} with components specified as following.
\begin{itemize}
    \item[1.] The wavelength dependent coefficients $\bdbeta(w_j)$ and eigenvector $\phi_k(w_j)$ in water and land are borrowed from a real data sample, latitude between 34.7129 and 35.7523 in orbit 05216 recorded on 06-25-2015. For our locations design, only the coefficients corresponding to footprint $p=4$ are used.
    \item[2.] For both land and water areas, the first principal component is assumed to be multivariate normal with covariance defined by exponential models (great circle distance $h$ in km): $G_1(h) = 5\exp(-h/10)$ for water area, and $G_1(h) = 10\exp(-h/7)$ for land area.
    \item[3.] For latter principal components, we assume $\xi_2^w(\bds_i) \overset{iid}{\sim} \mathcal{N}(0, 2)$, $\xi_2^l(\bds_i) \overset{iid}{\sim} \mathcal{N}(0, 2)$ and $\xi_3^l(\bds_i) \overset{iid}{\sim} \mathcal{N}(0, 1)$. %which are consistent with what we observed in the real data.
    %zhu consistent with real data?
    \item[4.] As observed in Fig \ref{fig:measerror_ftprint}, the standard error of measurement error at one footprint looks like a radiance function, which is further confirmed by our finding that it is proportional to the mean radiance within a small area. Let $\rho$ be the ratio, then
    \begin{align*}
        \sigma^t(w_j) = \rho \dfrac{1}{|\mathcal{S}_t|}\sum_{\bds_i \in \mathcal{S}_t} \mu(w_j; \bds_i), j = 1, \ldots, m_i,
    \end{align*}
    where $t = l, w$ and $p$ is dropped in the single footprint setting. The i.i.d. Lipschitz continuous stochastic process is generated as $
        e_i(w) = \nu_{i1} (1/\sqrt{c_w/2})\sin\{\pi (w - w_{min})/c_w\} 
     + \nu_{i2}(1/\sqrt{c_w/2})\cos\{\pi (w - w_{min})/c_w\}$, $w \in \mathcal{W}$, where $\nu_{i1}$ and $\nu_{i2}$ are independent and both are i.i.d. $\mathcal{N}(0,c_w/2)$, $c_w = |\mathcal{W}|$, $w_{min} = min\{w \in \mathcal{W}\}$.
    \item[5.] The land fraction $\alpha$ at the middle mixed point was generated from $\mathrm{Uniform}(0,1)$.
\end{itemize}

\begin{figure}[htbp]
    \centering
    \includegraphics[width = 0.7\textwidth]{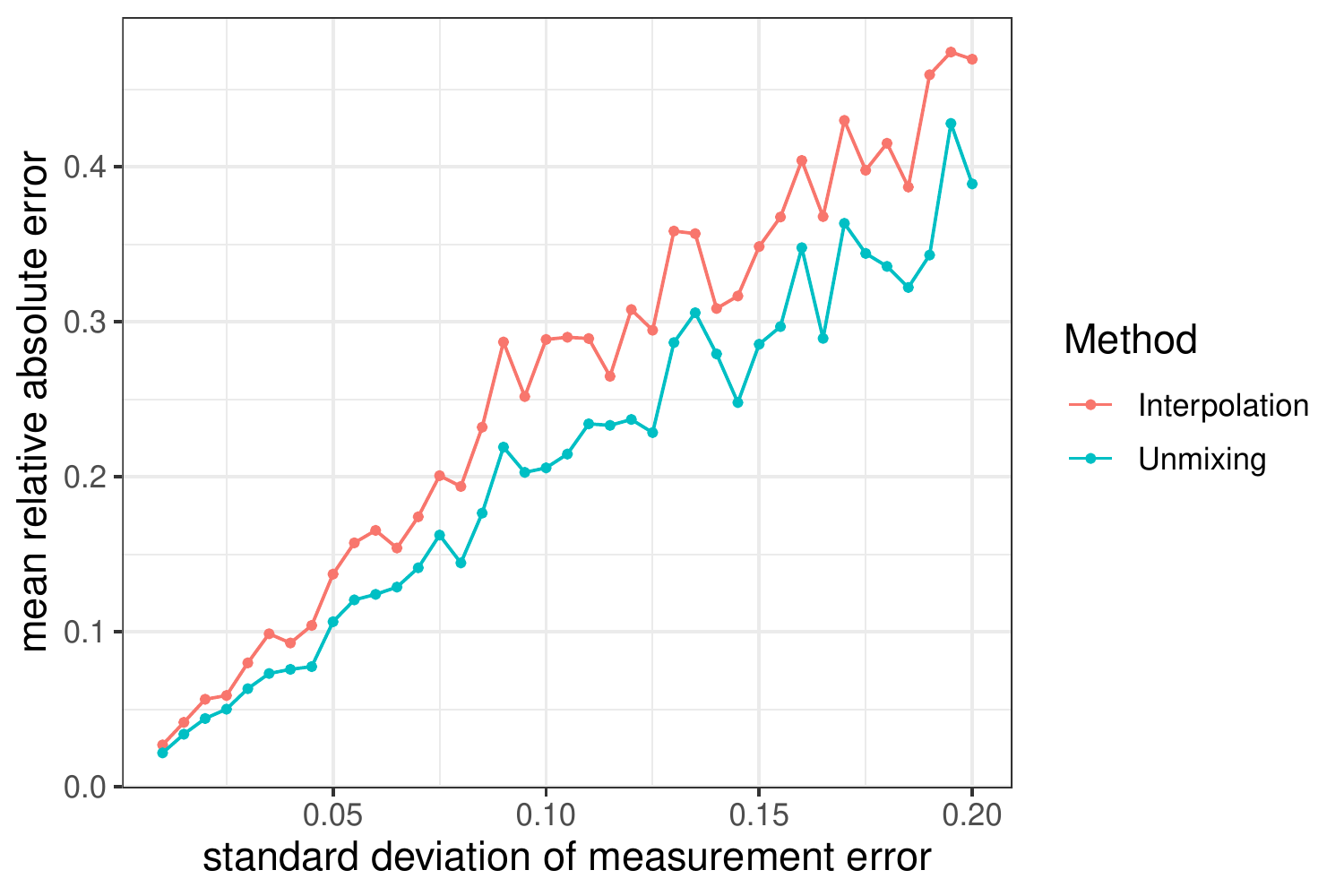}
    \caption{Relative absolute error against relative standard deviation of measurement error for both methods.}
\label{fig:unmixing_sim}
\end{figure}
Let $\wh \alpha_i$ be the land fraction estimator for location $\bds_i$, which can be either $\wh \alpha^t_i$ or $\wh \alpha^k_i$. The relative absolute error for the estimator $\wh \alpha_i$ is defined as $
    e_i = \dfrac{|\wh \alpha_i - \alpha|}{\alpha}
$.
We let relative standard deviation $\rho$, the ratio of the standard deviation to the mean radiance wavelength, vary from 0.01 to 0.2 and performed 200 simulations for each choice. Results are summarized from the 200 simulations by using 0.1 trimmed mean, which are shown in Fig \ref{fig:unmixing_sim}. 
The mean relative absolute error increases for both methods as data contains more noise, and unmixing approach based on FPCA and kriging is more stable and accurate than the simple linear interpolation across all levels of error. Based on our simulation results not shown here, benefits of unmixing approach over interpolation also become more dominant as the covariance of the measurement error process $e_i(w)$ becomes weaker.

\section{Theoretical Results} \label{asymptotics}
To derive asymptotic inference results for estimation, we first introduce notations here. In the rest of this paper, $||\cdot||_1$ denotes \textit{1-norm} and $||\cdot||_2$ denotes \textit{2-norm}. For a matrix or a vector $A$, $||A||_1 = \underset{||x||_1 = 1}{\max}||Ax||_1$ and $||A||_2 = \underset{||x||_2 = 1}{\max}||Ax||_2$. Given a sequences $f(n)$ and $g(n)$, the notation $f(n) = O(g(n))$ means there exists $c_1>0$ such that $|f(n)| \leq c_1 |g(n)|$, and $f(n) = \Omega(g(n))$ means that $|f(n)| \geq c_2 |g(n)|$ for some $c_2 > 0$. Also, $f(n) = \Theta(g(n))$ denotes when both $f(n) = O(g(n))$ and $f(n) = \Omega(g(n))$.

\subsection{Spatial Asymptotic Framework}
Based on the real data structures, we follow the spatial asymptotic framework in \cite{lahiri2003central} and adopt the mixed-increasing-domain structure under fixed design to develop the asymptotic results. Let $\mathcal{R}_0$ be an open subset of $(-1/2, 1/2]^2$ containing the origin, where two dimensions represent longitude and latitude in our context. Also, let $\lambda_n$ be a sequence of positive real numbers such that $\lambda_n \rightarrow \infty$ as $n \rightarrow \infty$. For each footprint $p$, we assume the sampling region as inflated $\mathcal{R}_0$ with a scaling factor $\lambda_n$ and a shift $(\Delta L_{n,p}, \Delta l_{n,p}) \in \mathbb{R}^2$. Formally, the sampling region for footprint $p$ is denoted as $\mathcal{R}_{np} = \lambda_n \mathcal{R}_0 + (\Delta L_{n,p}, \Delta l_{n,p})$. To avoid pathological cases, we follow \citep{lahiri2002asymptotic} and assume
\begin{enumerate}[label=(A.\arabic*)]
    \item \label{cond:sp_reg} For any sequence of positive real numbers $\{t_n\}$ with $t_n \rightarrow 0$ as $n \rightarrow \infty$, the number of cubes of the form $(\bdi + (0, 1]^2 )t_n, \bdi \in \mathbb{Z}^2$ that intersect both $\mathcal{R}_0$ and its complement $\mathcal{R}_0^c$ is $O(t_n)$ as $n \rightarrow \infty$.
\end{enumerate}
Thus, sampling region for each footprint has the same shape, but different centers. We define a lattice as $\mathcal{Z}^2 = \{(\delta_1 i_1, \delta_2 i_2): (i_1, i_2) \in \mathbb{Z}^2\}$, where $0 < \delta_1, \delta_2 < \infty$ are increments in two directions, and a positive sequence $h_n \rightarrow 0$ as $n \rightarrow \infty$ to vary the minimal distance in lattice. Under the \textit{mixed-increasing-domain} asymptotic framework, sampling locations at stage $n$ are intersection of scaled lattice $h_n \mathcal{Z}^2$, and the sampling region expanded at rate $\lambda_n$, i.e., 
\begin{align*}
    \mathcal{S}_{np}& = \mathcal{R}_{np} \cap (h_n \mathcal{Z}^2), \quad p = 1, 2, \ldots, 8\\
    \mathcal{S}_n & = \{\bds_1, \bds_2, \ldots, \bds_N\} = \cup_{p=1}^8 \mathcal{S}_{np},
\end{align*}
where the sample size $N$ and sample size $N_p$ in footprint $p$ need not be equal to $n$. Since our locations are deterministic, this is referred to as the \textit{fixed design} case. Summarizing, we require
\begin{enumerate}[label=(A.\arabic*), resume]
\item The sampling framework satisfies mixed-increasing-domain and fixed design.
\end{enumerate}
By (2.4) in \cite{lahiri2003central}, the fixed design in the mixed-increasing-domain case satisfies the growth condition,
$
\underset{n \rightarrow \infty}{\lim} r_0 \lambda_n^{2} h_n^{-2} / N_p = 1
$,
where $r_0$ is same positive constant depending on $R_0, \delta_1$ and $\delta_2$. As $n$ grows, the sample size increases at a rate of $\lambda_n^{2} h_n^{-2}$. \\
\textit{Remark:} The $\mathcal{S}$ and $\mathcal{S}_p$ defined above for estimation refers to the data sample actually used, which is considered as $\mathcal{S}_n$ and $\mathcal{S}_{np}$ at some certain stage $n$. Our real data setting shown in Fig \ref{fig:data_def} can be regarded as a special case of the \textit{mixed-increasing-domain} with \textit{fixed design} described above. The shift in our dataset is only in longitude since locations from different footprint are aligned along meridian. Also, the shape of sampling region $\mathcal{R}_0$ in our application is similar to a narrow parallelogram.

\subsection{Theoretical Results in Dense FPCA}
For any two subsets $\varLambda_1$ and $\varLambda_2$ of $\mathbb{R}^2$, $dist(\varLambda_1, \varLambda_2)$ is defined as $\inf \{||\bds_1 - \bds_2||_1: \bds_1 \in \varLambda_1, \bds_2 \in \varLambda_2\}$. We introduce the strong mixing coefficient $\alpha_k(a; b)$ \citep{guyon} to describe the dependence of the random field $\xi_k(\cdot)$, 
%Let $\bdxi_k = \{\xi_k(\bds_i): \bds_i \in \mathcal{S}\}$ be the collection of the $k$th principal component scores. 
%The strong mixing coefficient is defined as
\begin{align} \nonumber
\alpha_k(a; b)= & \sup \{ |P(A\cap B) - P(A) P(B)|: A \in \mathcal{F}_k(\varLambda_1), B \in \mathcal{F}_k( \varLambda_2), \\ \label{mix_coef}
 & |\varLambda_1| \leq b, |\varLambda_2| \leq b, dist(\varLambda_1, \varLambda_2) \geq a\},
\end{align}
where $\mathcal{F}_k(\varLambda)$ is the $\sigma$-algebra generated by the variables $\{\xi_k(\bds): \bds \in \varLambda\}$, $\varLambda \subset \mathbb{R}^2$.
%\xd{Please check the definition of the $\sigma$-algebra and the mixing coefficient. Is it generated by $\xi_k(\bds_i)$ or $\bdxi_k$?} it is generated by \xi_k(\bds_i)
To specify weak dependence via mixing coefficient, we assume \citep{lahiri2003central}
\begin{enumerate}[label=(A.\arabic*), resume]
    \item \label{cond:alpha} There exists a non-increasing function $\alpha_1(\cdot)$ with $\lim_{a \rightarrow \infty}\alpha_1(a) = 0$ and a non-decreasing function $g(\cdot)$ such that for $k = 1, \ldots, K$, $\alpha_k(a; b) \leq \alpha_1(a) g(b)$, $a > 0, b > 0$.
    \item \label{cond:alpha_1} $\int_{0}^{\infty} |y| \alpha_1(y) dy < \infty$.
    \item \label{cond:g} Define the function $f_1(t) = t^2 \int_{1}^t y^3 \alpha_1(y) dy, t \geq 1$, it satisfies that $g(t) = o(\{f_1^{-1}(t)\}^2/ \\ \{t\alpha_1(f_1^{-1}(t))\})$ as $t \rightarrow \infty$. 
\end{enumerate}
Conditions \ref{cond:alpha}--\ref{cond:g} state the restriction on the mixing coefficient, which can be satisfied under the following circumstances. For example, if $\alpha(a; b) \leq C (1+a)^{-\tau_1} b^{\tau_2}$, where $\tau_1 > 4$, $\tau_2 < \tau_1/2$ and $C > 0$, then \ref{cond:alpha_1} is satisfied with $f_1^{-1}(t) = \Theta(t^{1/2})$ and \ref{cond:g} is satisfied since $\{f_1^{-1}(t)\}^2/\{t\alpha_1(f_1^{-1}(t))\} = \Theta(t^{\tau_1/2})$. In addition, we need the following conditions for the geospatial functional data.
\begin{enumerate}[label=(A.\arabic*), resume]
    \item \label{cond:mm} For any $\bds_i \in \mathcal{S}_n$, $\E\left\{\underset{w \in \mathcal{W}}{\sup} |f(w; \bds_i)|^{c_1}\right\} < \infty$ for some $c_1 > 2$. Also, $\E\left\{\underset{w \in \mathcal{W}}{\sup}|e_i(w)|^{c_2}\right\} < \infty$ for some $c_2 > 4$, $|e_i(w) - e_i(w')| \leq L_i|w - w'|$ where $L_i$ are i.i.d. random variable with $\E L_i^4 < \infty$. 
    \item \label{cond:xi_k} There exists positive constant $C_0$ such that $|\xi_k(\bds)| < C_0$ a.s. for all $1 \leq k \leq K$.
    \item \label{cond:xi_k_smooth} For any $k$th principal component score within footprint $p$, it is continuous in probability, $|\xi_k(\bds + h) - \xi_k(\bds)| = O_p(h^{\beta})$ for some $\beta > 0$, $\bds \in \mathcal{S}_{np}$.
\end{enumerate}
The Theorem \ref{thm1} stated below establishes the uniform converge rate for $\wh \bdbeta(w)$ in (\ref{beta_estimate}) and $\wh \mu(w; \bds_i)$ in (\ref{mean_estimate}) for the proposed mean estimation.
\begin{thm}\label{thm1}
%Under mixed-increasing-domain asymptotic, fixed designs, and 
Under conditions \ref{cond:sp_reg}--\ref{cond:xi_k}, 
\begin{align*}
\sup_{w \in \mathcal{W}} ||\wh \bdbeta(w) - \bdbeta(w)||_2 & = O_p(1/\lambda_n), \\
\sup_{w \in \mathcal{W}} |\wh \mu(w; \bds_i) - \mu(w; \bds_i)| & = O_p(1/\lambda_n)
\end{align*}
for any $\bds_i \in \mathcal{S}_n$. 
\end{thm}

\noindent \textit{Remark:} This result shows that the converge rate of our mean function estimator is controlled by the expansion rate $\lambda_n$ of spatial locations. Using location coordinates in fitting linear model under mixed increasing domain, we obtain a rate at $1/\lambda_n \sim N^{-1/2} h_n^{-1}$. Compared to the root-$n$ convergence rate which is typical when variables are assumed to be independent or dependent but under the pure increasing domain asymptotics, this convergence rate is slower. As the infill density changes with rate $h_n^{-1}$, the increasing amount of dependence slows down the convergence rate of the mean estimates. 

%we simply use $\bds_{i-1}$ as $\bds_{i_1}$ and $\bds_{i+1}$ as $\bds_{i_2}$.
The next theorem shows the asymptotic property of covariance function estimates $\wh R_{\epsilon, p}(w, w')$ (\ref{meas_error_est}) and $\wh R_f(w, w')$ (\ref{cov_func_est}) under similar conditions. 
%\begin{align*} 
%\wh \sigma_p^2(w)  & = \dfrac{1}{6(N_p - 2)} \sum_{\substack{\bds_i \in \mathcal{S}_p}} \left\{r(w; \bds_{i-1}) - 2 r(w; \bds_{i}) + r(w; \bds_{i+1})\right\}^2.
%\wh \sigma^2(w_j) & = \dfrac{1}{n-1} \sum_{p = 1}^{8}n_p \wh \sigma_p^2(w_j)
%\end{align*}
\begin{thm} \label{thm2}
   %Under mixed-increasing-domain asymptotic, fixed designs, and 
   Under conditions \ref{cond:sp_reg}--\ref{cond:xi_k_smooth}, 
   \begin{align*}
   \sup_{w, w' \in \mathcal{W}} |\wh R_{\epsilon,p}(w, w') - R_{\epsilon,p}(w, w')| &= O_p(h_n/\lambda_n + h_n^{\beta_1}),\\
	\sup_{w, w' \in \mathcal{W}} |\wh R_f(w, w') - R_f(w, w')| &= O_p(1/\lambda_n + h_n^{\beta_1})
	\end{align*}
where $\beta_1 = \min\{1, \beta\}$, $\beta > 0$, and $p = 1, \ldots, 8$.
\end{thm}
\noindent \textit{Remark:} Similar to the mean parameter estimation, the convergence rate is related to how fast the area grow and how dense sampling points fill in the region. Compared to the mean function estimate, we also need measurement error variance estimate to converge using equation (\ref{meas_error_est}). Thus, the $h_n^{\beta_1}$ is included because of Condition \ref{cond:xi_k_smooth}, which shows that convergence depends on infill density and the continuity of principal component scores as a function of locations. 

The following result describes the consistency for estimating the principal component scores through dense functional principal component analysis (\ref{cov_decompose})--(\ref{tauk_estimate}).
\begin{thm} \label{thm3}
   %Under mixed increasing domain asymptotics, fixed sampling designs, and 
   Under Conditions \ref{cond:sp_reg}--\ref{cond:xi_k_smooth}, 
   \begin{align*}
  \underset{w \in \mathcal{W}}{\sup} |\wh \phi_k(w) - \phi_k(w)| & = O_p(1/\lambda_n + h_n^{\beta_1}), \\
   |\wh u_k(\bds_i) - u_k(\bds_i)| & = O_p(1/\lambda_n + h_n^{\beta_1}), \\
   |\wh \tau_k(p) - \tau_k(p)| & = O_p(1/\lambda_n + h_n^{\beta_1}),
   \end{align*}
where $\beta_1 = \min\{1, \beta\}$, $\beta > 0$, and $p = 1, \ldots, 8$.
\end{thm}
\noindent \textit{Remark:} By the fact that we are able to adopt numerical integral to estimate PC scores, it is easy to understand that Theorem \ref{thm3} provides the same converge rate as in Theorem \ref{thm2}. To reconstruct a radiance function in $\mathcal{S}$, we are able to have consistent estimate based on Theorem \ref{thm1}, \ref{thm2}, and \ref{thm3}.

\subsection{Theoretical Results for Variogram Estimation}
Following \cite{lahiri2002asymptotic}, we assume following regularity conditions on the semivariogram model $\gamma_k(h; \theta_k)$, $k = 1, \ldots, K$.
\begin{enumerate}[label=(C.\arabic*)]
    \item \label{c1} For any $\varepsilon > 0$ there exists a $\delta > 0$ such that $\inf\{\sum_{l=1}^L (\gamma_k(h_l; \theta_1) - \gamma_k(h_l; \theta_2))^2: ||\theta_1 - \theta_2||_2 \geq \varepsilon\} > \delta$.
    \item \label{c2} $\sup\{\gamma_k(h; \theta_k): h \in \mathbb{R}, \theta_k \in \Theta\} < \infty$, and $\gamma_k(h; \theta_k)$ is continuous with respect to $\theta_k$.
    %and $\gamma(h; \theta)$ has continuous partial derivatives of order s($\geq 0$) with respect to $\theta$.
\end{enumerate}
The exponential model is used as our choice of $\gamma_k(h; \theta_k)$, for which both \ref{c1} and \ref{c2} are satisfied on a compact parameter space for some distance lags $h_1, \ldots, h_L$. The weight matrix $V_k(\theta_k)$, $k = 1, \ldots, K$ is assumed to satisfy the following condition \citep{lahiri2002asymptotic}, 
\begin{enumerate}[label=(C.\arabic*), resume]
    \item \label{c3} $V_k(\theta_k)$ is positive definite for all $\theta_k \in \Theta$, $\sup\{||V_k(\theta_k)||_2 + ||V_k(\theta_k)^{-1}||_2: \theta_k \in \Theta\} < \infty$, and $V_k(\theta_k)$ is continuous on $\Theta$. 
    %$V(\theta)$ is r-times ($r \geq 0$) continuously differentiable on $\Theta$.
\end{enumerate}
Condition \ref{c3} requires the continuity of function $V_k(\theta_k)$ in $\Theta$, and common practices such as the use of diagonal matrix $N(h_l)/\gamma_k^2(h_l; \theta_k)$ \citep{cressie1985fitting}, $N(h_l)/h_l^2$ and $N(h_l)$ are all valid to use. The next theorem shows that the parameter estimator (\ref{thetak_estimate}) for variogram model is still consistent with the use of estimated scores. 
\begin{thm} \label{thm4}
   %Under mixed increasing domain asymptotics, fixed sampling designs, and 
   Under Conditions \ref{cond:sp_reg}--\ref{cond:xi_k_smooth}, \ref{c1}--\ref{c3},  
   $$\wh \theta_k - \theta_k \overset{p}{\rightarrow} 0 \quad as \quad n \rightarrow \infty.$$
\end{thm}
\noindent \textit{Remark:} By using Theorem 3.1 in \cite{lahiri2002asymptotic}, the proof is trivial since we have the consistency of $u_k(\bds_i)$ and $\tau_k(p)$ in Theorem \ref{thm3}. This result and Theorem \ref{thm3} are necessary to show that $\wh \xi_k(\bds_0)$ in (\ref{xik_predict}) converges to the BLUP via ordinary kriging.

\section{Applications to OCO-2 data} \label{application}
\subsection{Implementation Details}
In Section \ref{data}, we made it clear that the radiance is observed in a sequence of wavelengths and can be treated as a function in discrete index equivalently. Therefore, the domain $\mathcal{W}$ becomes $\{w \in \mathbb{Z}: 1\leq w \leq 1016\}$ for the sake of implementation of our cross-sectional estimation. In other words, this is a special case of the theoretical setting in Section \ref{model} and \ref{estimation}, which results in estimating mean, covariance and eigenfunction for a finite number of wavelengths. In practice, some of the OCO-2 spectral channels, in particular the largest and smallest wavelengths, do not produce scientifically reliable radiances. These ``bad samples'' are flagged and produce systematic patterns of missing radiances in the data product \citep{OCO2L1BATDB}. In this applicatioin we keep those wavelength indices with enough observations across locations, and focus on imputing radiance functions in wavelength indices available in the area of interest $\mathcal{S}$, which is denoted as $\mathcal{W}_a \subset \mathcal{W}, |\mathcal{W}_a| = m$.

Because the satellite orbit is monotonic (south-to-north) in latitude and the track is quite narrow, as shown in Fig \ref{fig:data_def}, it is sufficient to use latitude $L$ as the covariate in model (\ref{mean_model}). The FVE threshold we chose for determining the number of principal components in FPCA is 0.99. In real data, not all principal component scores contain spatial dependence, especially later ones explaining little fraction of variance. For ease of processing large samples efficiently, spatial dependence can be assessed with a permutation-based test \citep{cressie2015}. For components which are determined to have no spatial dependence, the BLUP estimator is reduced to $\wh \xi_k(\bds_0) = \dfrac{1}{N} \sum_{\bds_i \in \mathcal{S}} \wh u_k(\bds_i)$. Due to limited sample size, $\wh \gamma_k(h_l)$ may have negative values. For ordinary kriging, the \texttt{gstat} package is used in fitting the semivariogram models.

\subsection{Radiance Imputation over the Pacific Ocean} \label{imputation_app}

OCO-2 data aim to provide a comprehensive measurement framework for \ce{CO_2} concentration and the {\it{retrieval algorithm}} implements the estimation of $X_{CO2}$ from Level 1 data, high-resolution spectra of reflected sunlight.
As introduced in Section \ref{data}, OCO-2 has a large amount of locations with completely missing radiance because of atmospheric properties including clouds and cosmic rays. The spatial coverage of the retrieval algorithm will improve if the missing radiance can be imputed. 

\begin{figure}[ht]
	\begin{subfigure}{.48\textwidth}
	     \centering
		\includegraphics[trim = 2cm 0cm 2cm 0cm, clip = true, width = 0.8\textwidth]{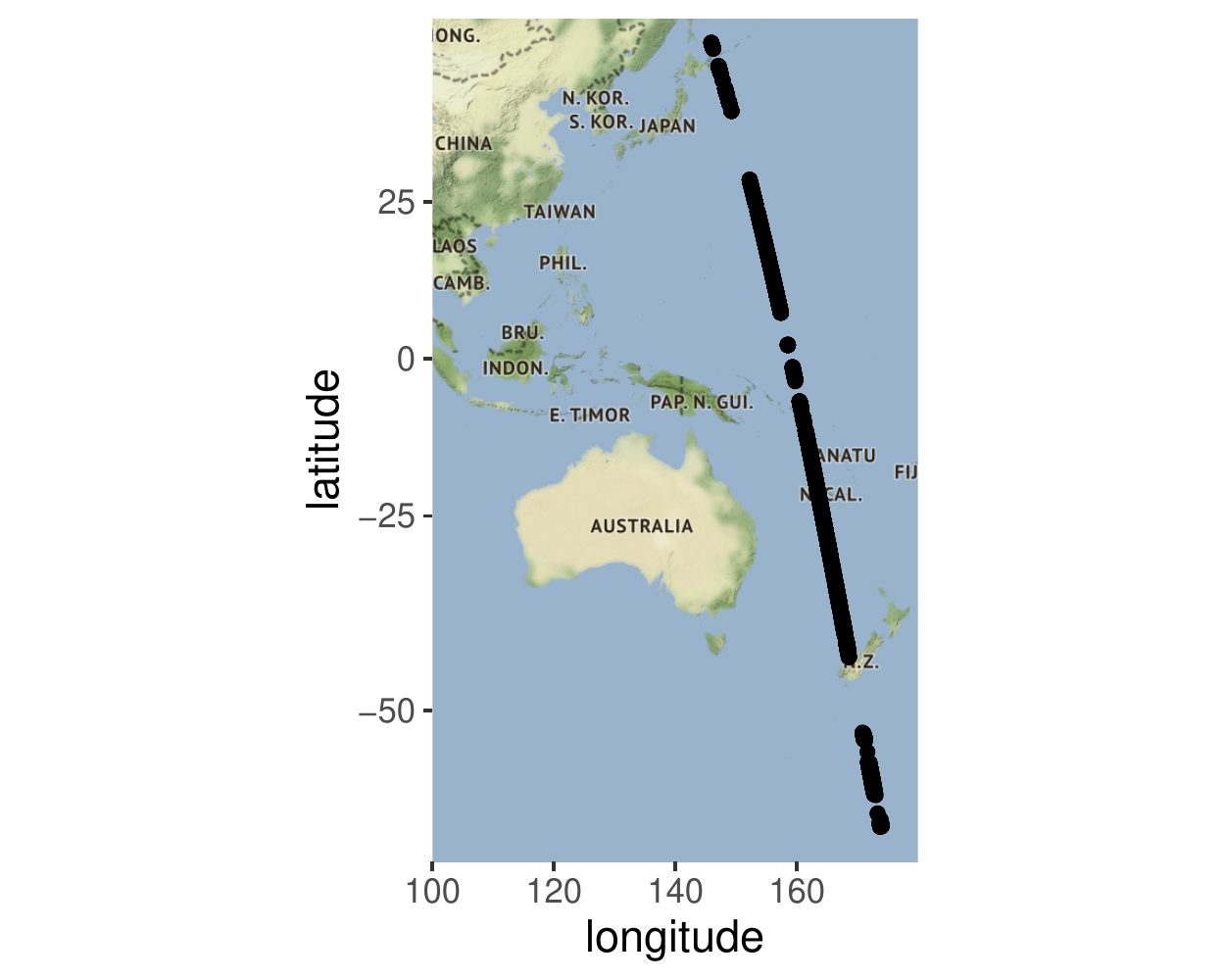}
	\end{subfigure}
	\begin{subfigure}{.48\textwidth}
		\centering
		\includegraphics[trim = 2cm 0cm 2cm 0cm, clip = true, width = 0.8\textwidth]{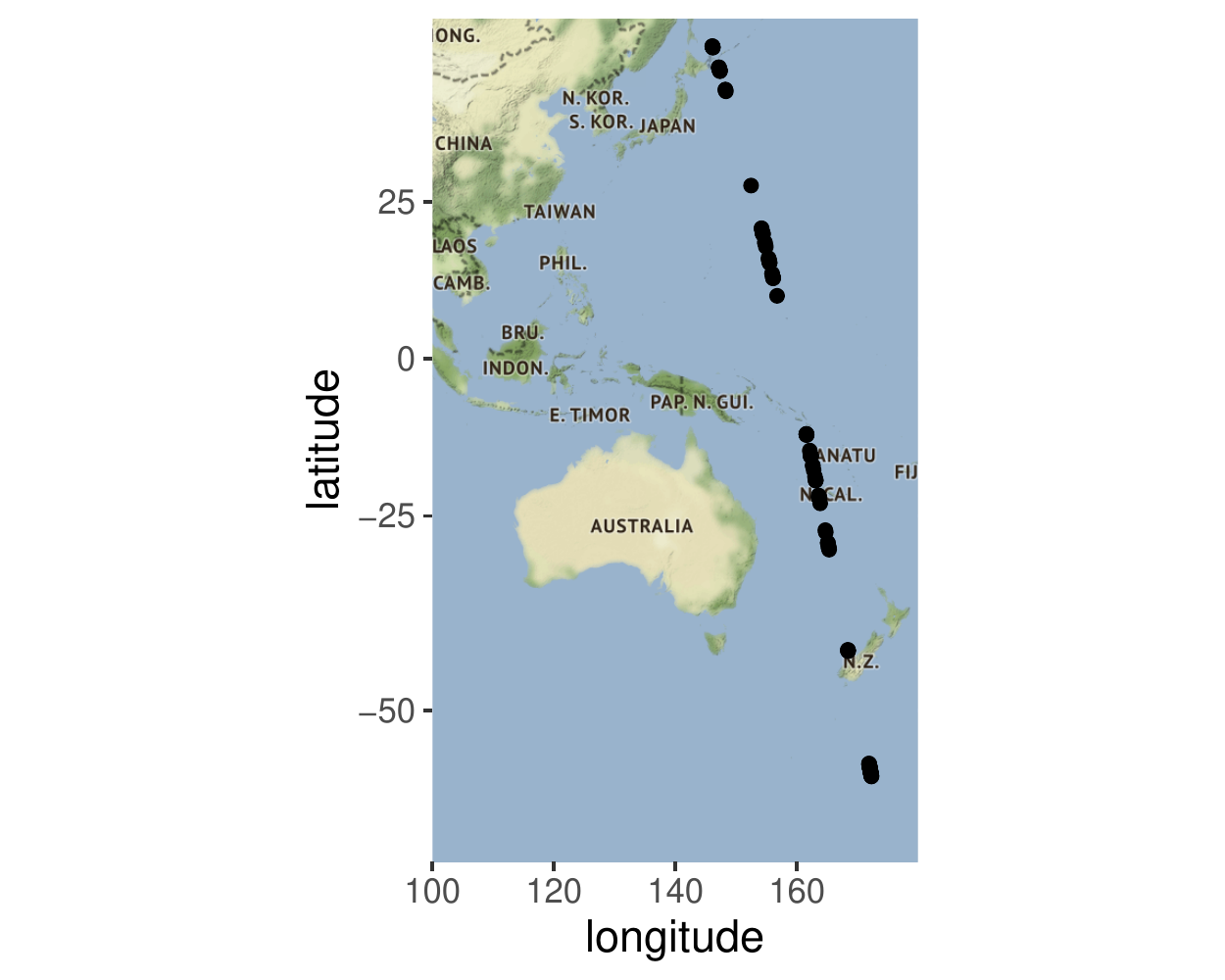}
	\end{subfigure}
	\caption{OCO-2 data used for radiance imputation validation. Left panel: orbit 14793 sounding on Pacific Ocean during 2017-04-13. Right panel: sampled 128 center points from orbit 14793 for experiment.}
	\label{fig:pacific_data}
\end{figure}
The data used as a case study in this section are downloaded from the OCO-2 Level 2 diagnostic data products \citep{v9dug}, available at the NASA Goddard Earth Science Data and Information Services Center (GES DISC, https://disc.gsfc.nasa.gov/OCO-2). As shown in Fig \ref{fig:pacific_data}, they are part of orbit 14793 in glint mode over the Pacific Ocean during 2017-04-13. The Level 1 variable {\tt{measured\_radiance}} was extracted from the dataset, along with geolocation information.

Preliminary examination indicates that locations within a region spanning about $0.5^{\circ}$ latitude in this area can be regarded as having homogeneous covariance function $R_f$. To illustrate our methods on radiance imputation in missing locations, we choose 128 locations at footprint 4, shown in Fig \ref{fig:pacific_data}, as the center points for doing the experiment described below. For each of the selected center points $\bds_i = (L_i, l_i)$, let $\mathcal{S}_i$ be the study area, which we define as the part of this orbit which has latitude between $L_i \pm 0.25^{\circ}$.
%zhu need clarification here, what is L?
We will take a segment of completely observed data around the center points out as the validation data, fit the model without the validation data, and compare the imputation with the validation data to evaluate the performance. Let $T_r(\bds)$ be the area consisted of the closest $r$ cross-tracks (1 to 8 footprints in a row, see Fig \ref{fig:data_def}) near $\bds$. For example, $T_1(\bds)$ would be the cross-track containing $\bds$. If $r$ is odd, $T_r(\bds)$ is the $T_1(\bds)$ plus $(r-1)/2$ cross-tracks below and above $\bds$. If $r$ is even, $T_r(\bds)$ is the $T_1(\bds)$ plus $r/2$ cross-tracks observed before $\bds$ and $(r/2-1)$ cross-tracks observed after $\bds$. In our validation study, we take $T_r(\bds)$ as the validation area, with $r$ ranging from 1 to 8.

To guarantee that we have enough data to fit our model after removing part of the data for validation, the center point sampled $\bds = (L, l)$ are selected to satisfy the following conditions.
\begin{itemize}
    \item[1.] The region between latitude $L \pm 0.25$ has at least 164 non-missing sounding locations such that we can have at least 100 points when $T_8(\bds)$ is removed.
    \item[2.] $T_8(\bds)$ does not have missing locations, i.e., there is no gaps and total number of observations is 64.
\end{itemize}

%the local linear smoothing functions $\wt f(w, \cdot), w \in \mathcal{W}_a$ for each point in the 8 by 8 grid $T_8(\bds)$ are treated as true radiance to compare with imputed radiance. 
We conducted the following procedure to validate our imputation algorithm, and repeats it for all 128 selected center points. For $r = 1, \ldots, 8$, 
\begin{itemize}
    \item[1.] Remove the region defined as $T_r(\bds)$, i.e., the $r$ nearest cross-tracks around $\bds$. 
    \item [2.] Impute radiance function for the removed area $T_r(\bds)$ using the radiance data left in the selected validation region $\mathcal{S}$.
    \item[3.] Calculate the RRMSE (Root Relative Mean Squared Error) 
    %and RMSPE (Root Mean Squared Prediction Error) 
    for the imputed sounding locations $\{\bds_0: \bds_0 \in T_r(\bds)\}$, which is defined as
    \begin{align}
    \label{rmse_eqn}
    \text{RRMSE} & = \sqrt{\dfrac{1}{m} \sum_{w \in \mathcal{W}_a} \dfrac{\{\wh f(w; \bds_0) - r(w; \bds_0)\}^2}{r^2(w; \bds_0)}},
    %\label{rpmse_eqn}
    %\text{RMSPE} &  = \sqrt{\dfrac{1}{m} \sum_{w \in \mathcal{W}_a} \left[\sum_{k=1}^K \{u_k(\bds_0) - \wh \xi_k(\bds_0)\} \wh \phi_k (w)\right]^2},
    \end{align}
    %where $u_k(\bds_0) = \int \{r(w; \bds_0) - \wh \mu(w; \bds_0)\} \wh \phi_k(w) dw$
\end{itemize}
%zhu RMSE should be RRMSE (root relative mean squarte error), what is RPMSE? need to spell it out

Since measured radiance can be all missing for some specific wavelength and footprint, we can either fill in missing radiance by interpolation or discard the wavelength directly such that $\bdbeta(w)$ can be fully estimated. In this paper, we present the latter approach as they do not differ in terms of imputation performance. Considering radiance at each wavelength changes smoothly along geospatial locations, we treat a naive linear interpolation method as the benchmark to compare with our proposed functional approach. For a location $\bds$ with missing observations, the radiance $f(w; \bds)$ can be imputed by linear interpolation using measured radiance of wavelength $w$ observed at locations with footprint $p = q(\bds)$. The experiment procedure above was also applied to validate linear interpolation with RRMSE (\ref{rmse_eqn}) as the evaluation criteria.

Aggregating results by number of cross-tracks removed, we can see how the radiance imputation method perform when the area to be imputed changes. For each number of cross-tracked removed, we calculate the average RRMSE over all imputation results, and large sample 95\% confidence interval for the mean. Then as shown in Fig \ref{fig:rmse_compare}, RRMSE increases as number of cross-tracks removed increases for both functional approach and linear interpolation. However, the proposed method based on our geospatial functional model consistently attains a lower RRMSE than linear interpolation. Furthermore, the advantage over linear interpolation becomes significant when the missing area is larger than 5 cross-tracks. 
\begin{figure}[htbp]
    \centering
\begin{subfigure}[t]{0.49\textwidth}
    \centering
    \includegraphics[trim = 0 0.5cm 0 0, clip = true, width = \textwidth]{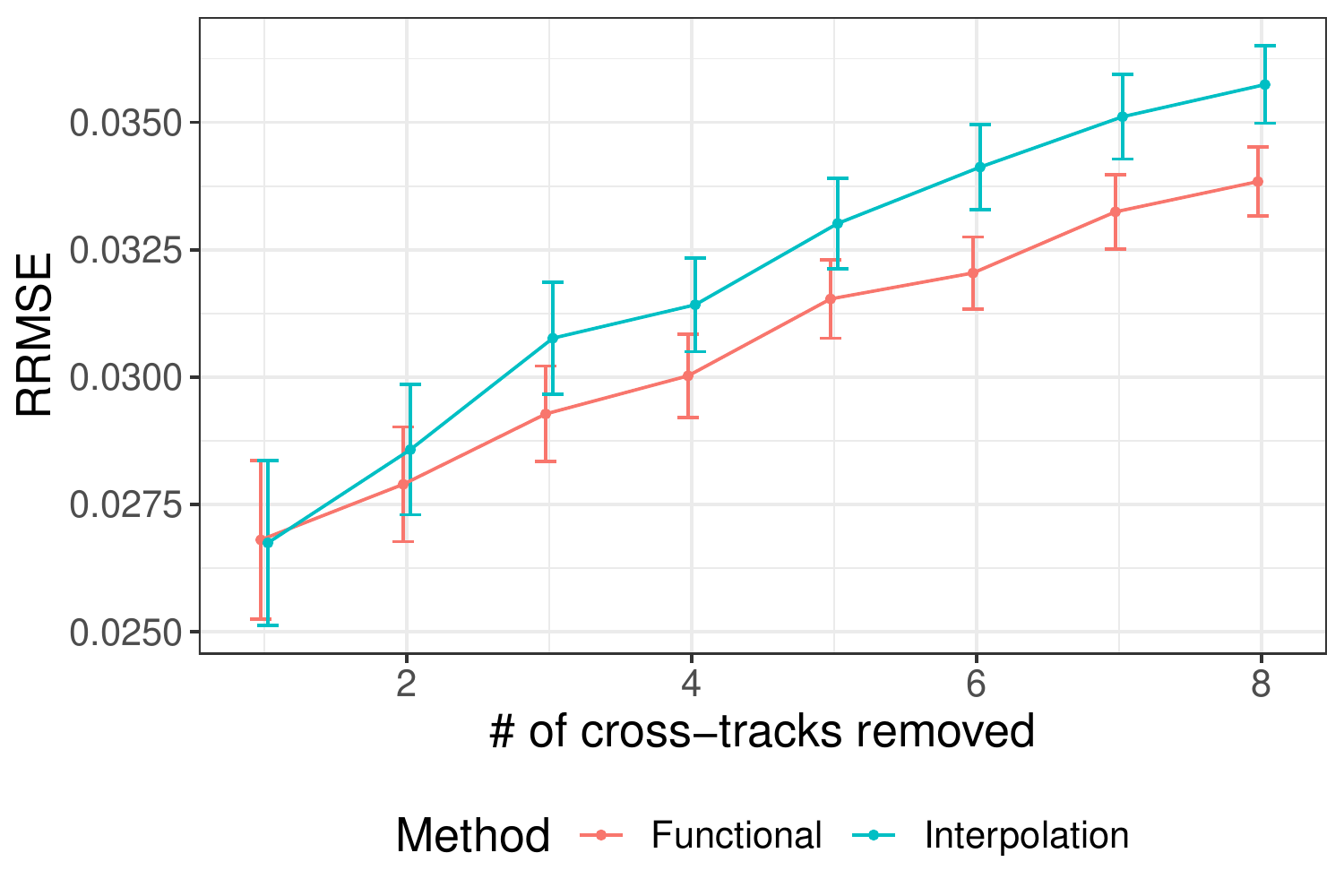}
    \caption{}
    \label{fig:rmse_compare}
\end{subfigure}
\begin{subfigure}[t]{0.49\textwidth}
    \centering
    \includegraphics[trim = 0 0.5cm 0 0, clip = true, width =\textwidth]{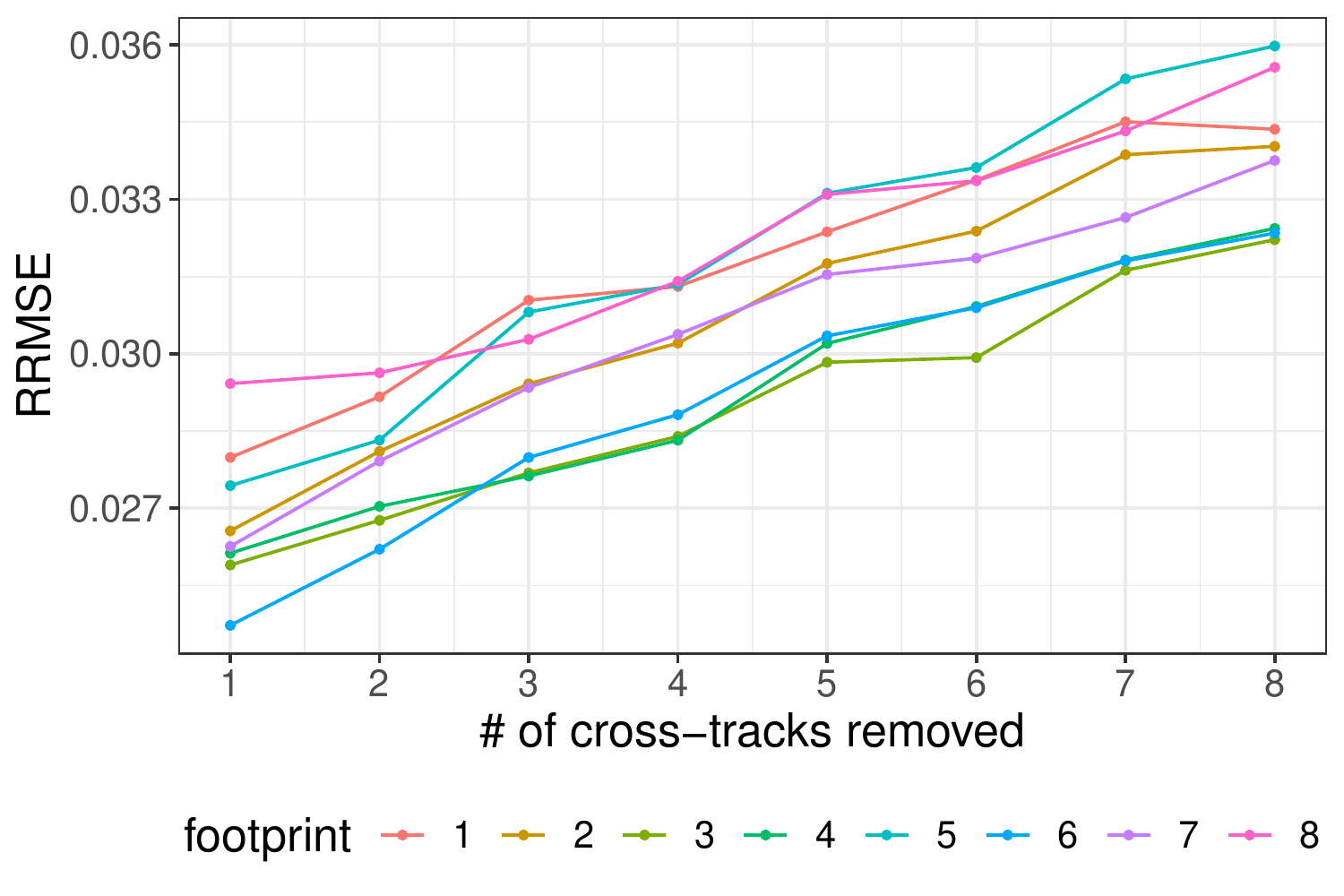}
    \caption{}
   \label{fig:rmse_rpmse_plot}
\end{subfigure}
\caption{(a) Average and 95\% confidence intervals of RRMSE with respect to the number of cross-track removed, \textit{red}: proposed functional approach, \textit{green}: linear interpolation. (b) Average RRMSE in all implementations for different footprints, against 1--8 number of cross-tracks removed.}
\end{figure}
%\begin{figure}[htbp]
%\begin{subfigure}{1\textwidth}
%    \centering
%    \includegraphics[trim = 0cm 3cm 0cm 3cm, clip = true, width = 0.95\textwidth]{14793_rmse_1-4.png}
%    \caption{average RRMSE in all implementations for 1-4 cross-tracks removed}
%    \label{fig:rmse1}
%\end{subfigure}
%
%\begin{subfigure}{1\textwidth}
%    \centering
%    \includegraphics[trim = 3cm 0cm 3cm 0cm, clip = true, width = 0.95\textwidth]{14793_rmse_5-8.png}
%    \caption{average RRMSE in all implementations for 5-8 cross-tracks removed}
%    \label{fig:rmse2}
%\end{subfigure}
%\caption{Removed cross-tracks colored by RRMSE in radiance imputation.}
%\end{figure}

In addition to showing proposed approach is overall better than the naive interpolation method, we want to see in detail how the imputation performance differs across footprints. Fig \ref{fig:rmse_rpmse_plot} shows for each footprint, the average RRMSE against the number of cross-tracks removed, respectively. It is consistent with the result in Fig \ref{fig:rmse_compare}: the RRMSE shows an increasing trend, since imputation gets harder when missing region size increases. As lying on the boundaries, imputation for footprint 1 and 8 are generally worse than the other footprints, except for footprint 5 having the highest RRMSE when number of cross-tracks removed is more than 5.
%The plot of RMSPE shows that footprint 5 is the worst in terms of predicting component scores.
To see average imputation performance over 128 experiments in each missing region $T_r(\cdot), r = 1, \ldots, 8$, please find heatmaps in Supplementary Material S.2 for a more comprehensive representation. 

%\begin{figure}[htbp]
%\begin{subfigure}{0.48\textwidth}
    %\centering
    %\includegraphics[width =0.8\textwidth]{rmse_gap_ctk_plot}
    %\caption{Average RRMSE in all implementations on the cross-track of center points for different footprints}
    %\label{fig:rmse_plot}
%\end{subfigure}
%\begin{subfigure}{0.48\textwidth}
%    \centering
%    \includegraphics[width = \textwidth]{rpmse_gap_ctk_plot}
    %\label{fig:rpmse_plot}
%\end{subfigure}
%\caption{Average RRMSE in all implementations for different footprints, against 1--8 number of cross-tracks removed.}
%\label{fig:rmse_rpmse_plot}
%\end{figure}

%\begin{figure}[htbp]
%\begin{subfigure}{1\textwidth}
%    \centering
%    \includegraphics[trim = 0cm 3cm 0cm 3cm, clip = true, width = 0.95\textwidth]{14793_rpmse_1-4.png}
%    \caption{average RMSPE in all implementations for 1-4 cross-tracks removed}
%    \label{fig:rpmse1}
%\end{subfigure}

%\begin{subfigure}{1\textwidth}
%    \centering
%    \includegraphics[trim = 3cm 0cm 3cm 0cm, clip = true, width = 0.95\textwidth]{14793_rpmse_5-8.png}
%    \caption{average RMSPE in all implementations for 5-8 cross-tracks removed}
%    \label{fig:rpmse2}
%\end{subfigure}
%\caption{Removed cross-tracks colored by RMSPE in radiance imputation.}
%\end{figure}

\subsection{Land Fraction Correction around Greece}
The land fraction variable $\alpha_i$ in OCO-2 data is computed by mapping the OCO-2 location (longitude/latitude) to a static land/water mask. Due to the geo-location uncertainties and the static nature of the mask in this procedure, the land fraction value provided in the OCO-2 data is not reliable. Our unmixing approach can be used to provide a more accurate land fraction estimates for the retrieval algorithm and further increase its spatial coverage.

To evaluate our land fraction estimation method, we use data provided in the OCO-2 Level 1B product \citep{OCO2L1BATDB} along the coastal area of Greece, which includes orbit 05449 on 07-11-2015 and orbit 05216 on 06-25-2016. As shown in Fig \ref{unmixing_data}, there are four mixed regions to estimate respectively after swiping out wavelength index with missing observations. Since the satellite has a repeat cycle of 16 days, orbit 05449 and 05216 are actually in the same area, though they do not have exactly the same coordinates.

Since we do not know the true land fraction, ground truth data regarding the selected two orbits were created manually. The coastal satellite pictures were downloaded from Google Earth and coastlines were added by feature editing in ArcGIS. Then each footprint's land fraction was obtained by calculating the proportion of the coastline polygon inside the area constructed by its four vertices.
\begin{figure}[htbp]
\begin{subfigure}{.48\textwidth}
    \centering
    \includegraphics[trim = 0cm 0cm 0cm 0cm, clip = true, width = \textwidth]{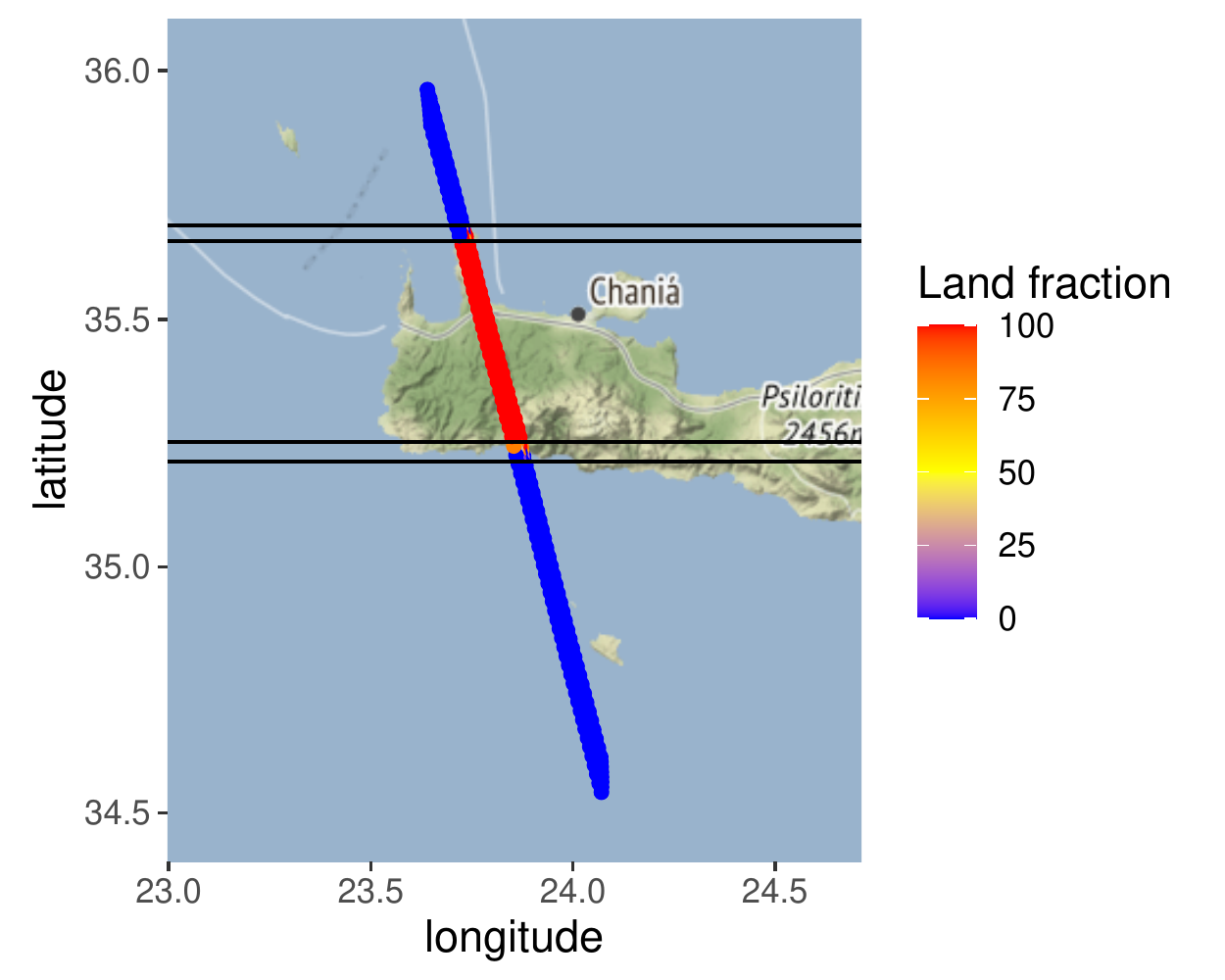}
\end{subfigure}
\begin{subfigure}{.48\textwidth}
    \centering
    \includegraphics[trim = 0cm 0cm 0cm 0cm, clip = true, width = \textwidth]{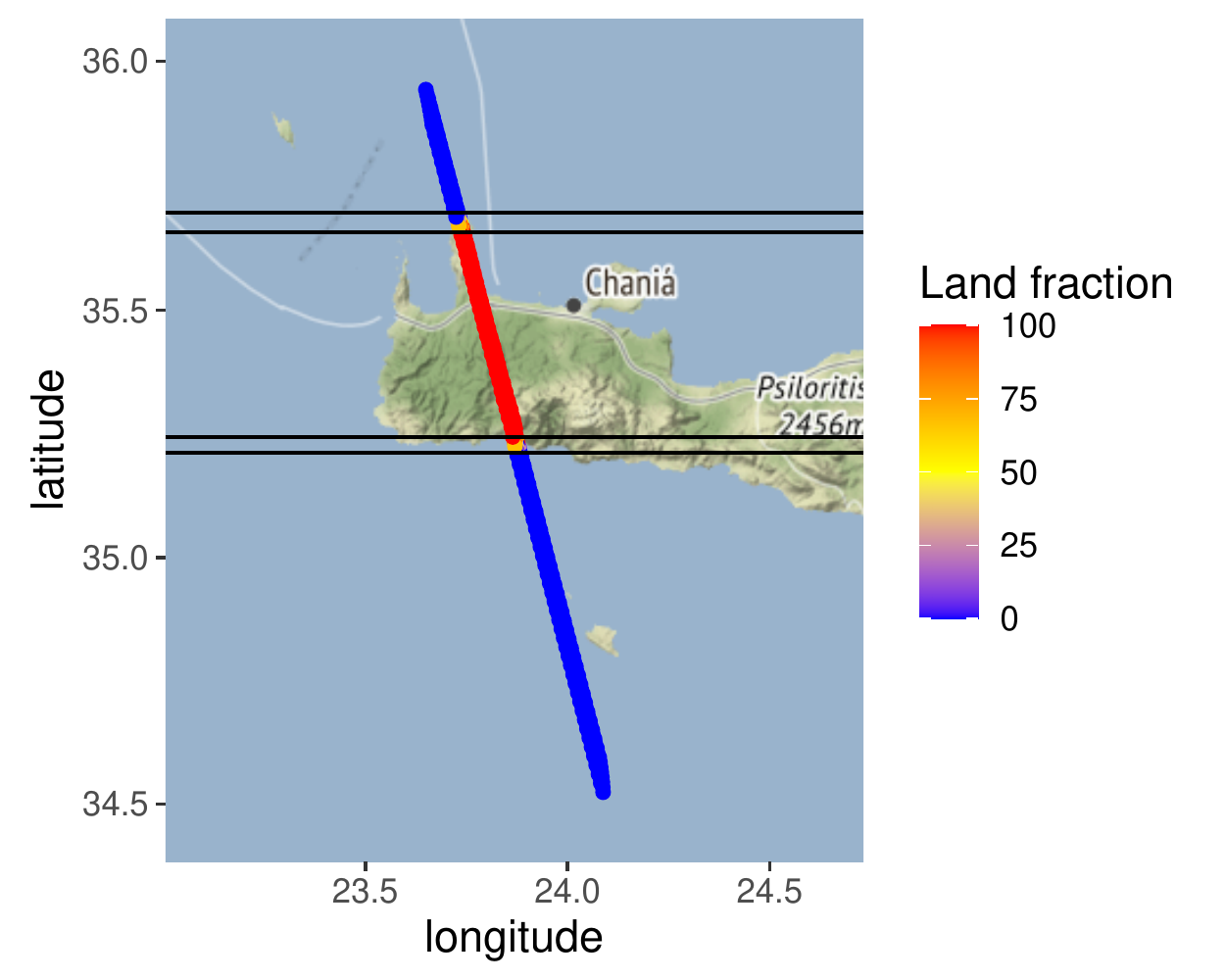}
\end{subfigure}
\caption{OCO-data used for land fraction estimation: Orbit 05216 (left panel) and Orbit 05449 (right panel). Points are colored according to the reported land fraction for each sounding. }
\label{unmixing_data}
\end{figure}

A mixed region is defined as the transition zone from either land to water or water to land. Typically some mixed land/water soundings are present at transition zones. In our algorithm, land fractions are to be estimated in area $\mathcal{M} = \{\bds_i: L_1 - \delta_0 < L_i < L_2 + \delta_0 \}$ where $L_1$ and $L_2$ are the minimum and maximum latitude of the mixed region, and $\delta_0$ is the tolerance (usually set as the average latitude difference) made to account for effect of bad land fraction on the area surrounding the water/land regions. We choose the lower unmixed region as $\mathcal{S}_1 = \{\bds_i: L_1 - \delta_0 - 0.6 \leq L_i \leq L_1 - \delta_0\}$, and the upper unmixed region as $\mathcal{S}_2 = \{\bds_i: L_2 + \delta_0 \leq L_i \leq L_2 + \delta_0 + 0.6\}$, to provide enough data for model estimation. For a given region $\mathcal{M}$, the algorithm conducts the following procedure.
\begin{itemize}
	\item[1.] Determine the type of $\mathcal{S}_1$ and $\mathcal{S}_2$ by land fraction average: for $t = 1, 2$, compute 
	$\sum_{\bds_i \in \mathcal{S}_t} \alpha_i / |\mathcal{S}_t|$. 
	It is recognized as land if the average is more than 70\%, and recognized as water if the average is less than 30\%. Otherwise it is labeled as unidentified. The data is qualified for our unmixing approach if both land and water are identified.
	\item[2.] Do spectral imputation for locations in $\mathcal{M}$ using radiance of unmixed locations in area $\mathcal{S}_1$ and $\mathcal{S}_2$ as input separately. The imputation algorithm is the same as what we proposed above (Section \ref{imputation_app}) except that outliers are removed and a local linear smoother is applied on $\wh \bdu_{k}$ to reduce large variation around $\mathcal{M}$. The bandwidth is selected by cross validation or fixed at 0.1 if any extreme $\wh u_{k}(\bds_i)$ near region $\mathcal{M}$ is detected.
	\item[3.] Estimate $\alpha_i$ for $\bds_i \in \mathcal{M}$ by (\ref{unmixing}), and truncated between 0 and 1.
\end{itemize}
\begin{table}[ht]
	\centering
	\caption{Mean squared error for unmixng and OCO-2 estimates: $\sum_{\bds_i \in \mathcal{M}} (\wh \alpha_i - \alpha_i)^2 / |\mathcal{M}|$}
	\begin{tabular}{ccccc}
		\hline
		MSE &  05216 lower & 05216 upper & 05449 lower & 05449 upper\\
		\hline 
		Unmixing estimates & 0.00300 & 0.07426 & 0.04662 & 0.00783 \\
		%Nearest & 0.00136 & 0.09824 & 0.06813 & 0.01903 \\
		OCO-2 & 0.05916 & 0.09546 & 0.08448 & 0.06434 \\
		\hline 
	\end{tabular}
	\label{land_fraction_table}
\end{table}
\begin{figure}[ht]
    \centering
    \includegraphics[trim = 0 0.5cm 0 0.5cm, clip = true, width = 0.8\textwidth]{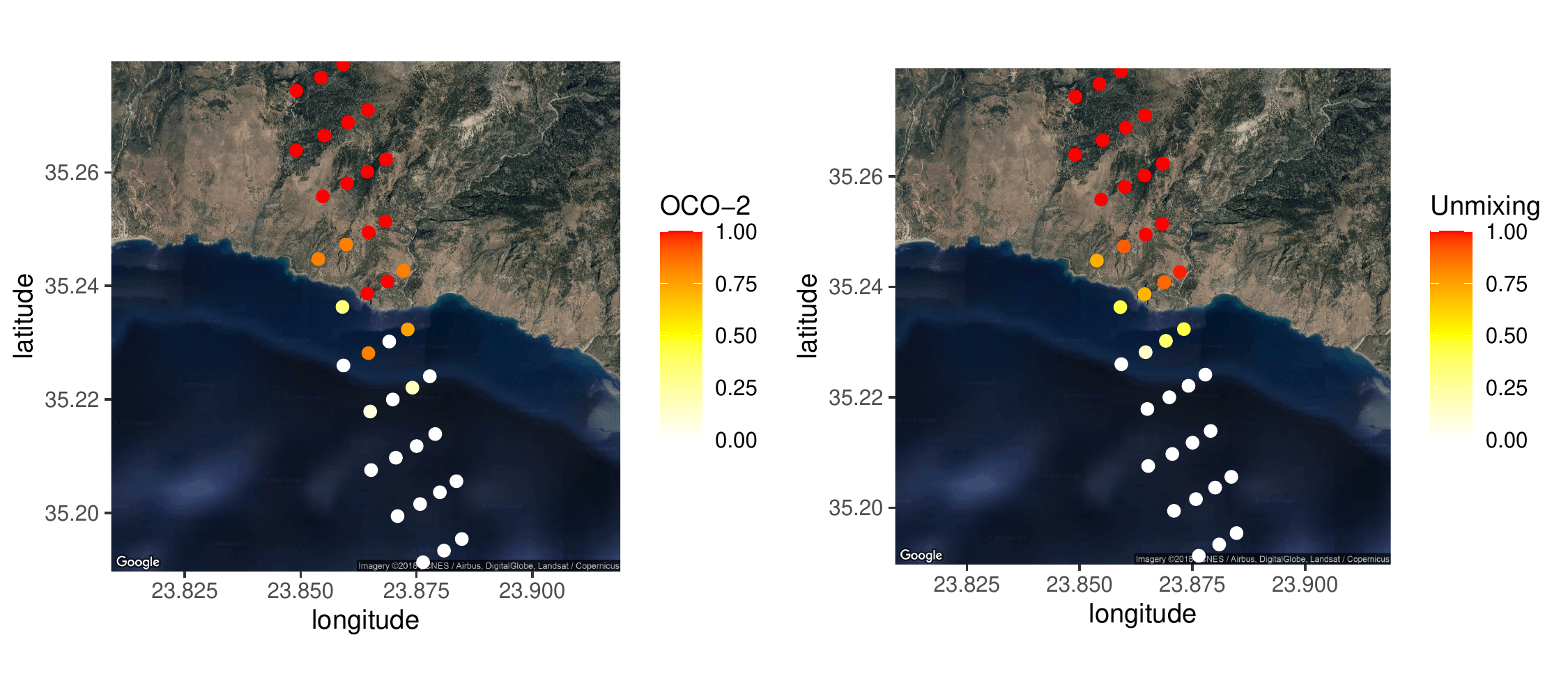}
    \caption{Land fraction heat map of 05126 orbit's lower area: OCO-2 estimate on the left, Unmixing estimate on the right.}
    \label{land_fraction_heatmap}
\end{figure}

We implemented the unmixing algorithm on four mixed regions in Fig \ref{unmixing_data}, and the numerical results are summarized in Table \ref{land_fraction_table}, and Fig \ref{land_fraction_plot} shows land fraction estimates with respect to latitude in these mixed regions. In all four cases the unmixing approach gives a substantially more accurate and reliable land fraction estimate compared to the original estimates in OCO-2 data. The MSE of the unmixing estimates is very small compared to the MSE of OCO-2 land fractions for the lower region of 05216 and the upper region of 05449. The improvement on land fraction estimation is clearly in Fig \ref{land_fraction_heatmap}, where the color change is much smoother and more reasonable in heatmap right-hand side.
\begin{figure}[ht]
\centering
\includegraphics[width = 0.85\textwidth]{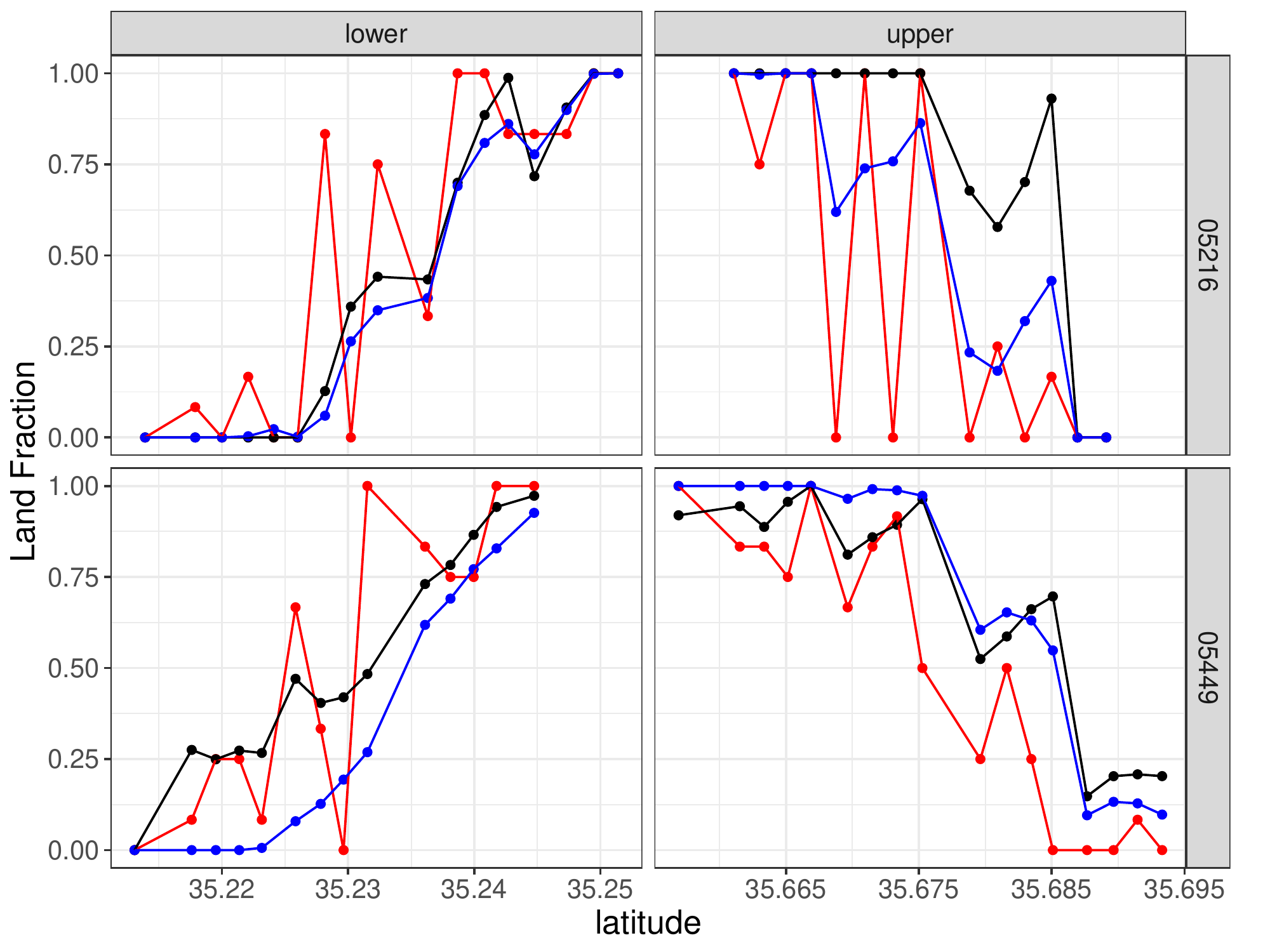}
\caption{Land fraction from three resources against latitude in four mixed regions. \textit{black}: unmixing estimates, \textit{red}: estimates in OCO-2 data, \textit{blue}: ground truth created manually.}
\label{land_fraction_plot}
\end{figure}
Due to large variation near mixed region, the improvements in MSE in the upper area of 05216 and lower area of 05449 is more moderate. 
However, a closer inspection of upper-right and bottom-left panels in Fig \ref{land_fraction_plot} show that our unmixng estimate (black line) is much more aligned with the ground truth (blue line) compared to OCO-2 results (red line). In general, in all of the four mixed regions, OCO-2 estimate is the unstable and unrealistic, while unmixing estimate stays more consistent with ground truth.
%\begin{table}[H]
%	\centering
%	\caption{Mean squared error for unmixng and OCO-2 estimates: $\sum_{\bds_i \in \mathcal{M}} (\wh \alpha_i - \alpha_i)^2 / |\mathcal{M}|$}
%	\begin{tabular}{ccccc}
%		\hline 
%		MSE &  05216 lower & 05216 upper & 05449 lower & 05449 upper\\
%		\hline 
%		Unmixing estimates & 0.0056 & 0.0582 & 0.0132 & 0.0053 \\
%		OCO-2 & 0.0592 & 0.0955 & 0.0845 & 0.0643 \\
%		\hline 
%	\end{tabular}
%	\label{land_fraction_table}
%\end{table}

\section{Discussion} \label{discussion}
In this paper, we introduced a geospatial functional model for spatial spectral data to impute missing hyperspectral radiance data from the OCO-2 satellite. The model treats spectral radiance as a function of wavelength, and introduces spatial dependence among radiance function by modeling the FPC scores as spatial processes. The model allows footprint-specific mean radiance functions and measurement error process to account for the observed heterogeneity across footprints. Practical algorithms were developed for parameter estimation and prediction, and asymptotic consistency and convergence rate results for the procedure were established.
%This unified modeling framework allows accurate imputation and land fraction estimation. 
%is able to account for principal component with and without spatial dependence, which means imputation is robust and flexible to globally varying radiances. 
We successfully implemented the algorithm and achieved acceptable high accuracy for radiance imputation at footprints over water. Furthermore, we developed an unmixing approach to estimate the land fraction in mixed footprints based on the imputation method, and shown that it gives much more accurate land fraction then those provided in the OCO-2 data. Both the imputation and the estimation of land fraction can help increase the spatial coverage of the OCO-2 retrieval algorithm and improve its utility to the research community studying the sources and sinks of carbon dioxide, and contributing to scientists' understanding of how carbon is contributing to climate change. 

The proposed model and algorithm for doing imputation and land fraction estimation is effective over a homogeneous area. Our observation is that for OCO-2 data it works well within a latitude range of about $0.6^{\circ}$ . Due to non-stationarity and weakened spatial dependence over longer distance, this model's effectiveness is diminished for filling large gaps in the data. In this case, a more complex space-time covariance model for spatial spectral data can be developed and studied to tackle imputation at a global scale. This will consider in a separate paper. The methodology developed in this paper also has the potential utility in data processing for hyper-spectral satellite data similar to OCO-2. 

% The imputation approach could allow for additional successful Level 2 retrievals when radiances can be successfully imputed.  Currently the OCO-2 operational Level 2 retrieval algorithm has slightly different configurations for land and ocean soundings. Retrievals are not attempted for mixed land/ocean soundings \citep{odell2018}. An accurate estimate of land fraction from the radiance data could facilitate retrievals in these mixed cases. The unmixed land fraction estimates may also have sufficient accuracy to supplement OCO-2's geolocation information.

%\bigskip
%\begin{center}
%{\large\bf SUPPLEMENTARY MATERIAL}
%\end{center}

%\begin{description}

%\item[Title:] Brief description. (file type)
%
%\item[R-package for  MYNEW routine:] R-package ÒMYNEWÓ containing code to perform the diagnostic methods described in the article. The package also contains all datasets used as examples in the article. (GNU zipped tar file)

%\item[HIV data set:] Data set used in the illustration of MYNEW method in Section~ 3.2. (.txt file)

%\end{description}
\spacingset{1}

\bibliographystyle{apalike}

\bibliography{oco2_ref.bib}

\begin{thebibliography}{}

\bibitem[Bradley et~al., 2005]{bradley2005basic}
Bradley, R.~C. et~al. (2005).
\newblock Basic properties of strong mixing conditions. a survey and some open
  questions.
\newblock {\em Probability surveys}, 2:107--144.

\bibitem[Castro et~al., 1986]{castro1986}
Castro, P.~E., Lawton, W.~H., and Sylvestre, E. (1986).
\newblock Principal modes of variation for processes with continuous sample
  curves.
\newblock {\em Technometrics}, 28(4):329--337.

\bibitem[Cressie, 1985]{cressie1985fitting}
Cressie, N. (1985).
\newblock Fitting variogram models by weighted least squares.
\newblock {\em Journal of the International Association for Mathematical
  Geology}, 17(5):563--586.

\bibitem[Cressie, 2018]{cressiejasa}
Cressie, N. (2018).
\newblock {Mission \ce{CO2}ntrol: A statistical scientist's role in remote
  sensing of carbon dioxide}.
\newblock {\em Journal of the American Statistical Association}, 113:152--168.

\bibitem[Cressie and Wikle, 2015]{cressie2015}
Cressie, N. and Wikle, C.~K. (2015).
\newblock {\em Statistics for spatio-temporal data}.
\newblock John Wiley \& Sons.

\bibitem[Crisp et~al., 2017]{CrispCal}
Crisp, D., Pollock, H.~R., Rosenberg, R., Chapsky, L., Lee, R. A.~M., Oyafuso,
  F.~A., Frankenberg, C., O'Dell, C.~W., Bruegge, C.~J., Doran, G.~B.,
  Eldering, A., Fisher, B.~M., Fu, D., Gunson, M.~R., Mandrake, L., Osterman,
  G.~B., Schwandner, F.~M., Sun, K., Taylor, T.~E., and Wunch, D. (2017).
\newblock {The on-orbit performance of the Orbiting Carbon Observatory-2
  (OCO-2) instrument and its radiometrically calibrated products}.
\newblock {\em Atmospheric Measurement Techniques}, 10:59--81.

\bibitem[Delicado et~al., 2010]{delicado2010}
Delicado, P., Giraldo, R., Comas, C., and Mateu, J. (2010).
\newblock Statistics for spatial functional data: some recent contributions.
\newblock {\em Environmetrics: The official journal of the International
  Environmetrics Society}, 21(3-4):224--239.

\bibitem[Eldering et~al., 2017a]{ElderingEtal2017}
Eldering, A., O'Dell, C.~W., Wennberg, P.~O., Crisp, D., Gunson, M., Viatte,
  C., et~al. (2017a).
\newblock {The Orbiting Carbon Observatory-2: First 18 months of science data
  products}.
\newblock {\em Atmospheric Measurement Techniques}, 10(2):549--563.

\bibitem[Eldering et~al., 2017b]{OCO2L1BATDB}
Eldering, A., Osterman, G., Pollock, R., Lee, R., Rosenberg, R., Oyafuso, F.,
  Crisp, D., Chapsky, L., and Granat, R. (2017b).
\newblock {\em {Orbiting Carbon Observatory (OCO-2)} {Level 1B} Algorithm
  Theoretical Basis}.
\newblock Jet Propulsion Laboratory.
\newblock {JPL} document {OCO} {D}-55206.

\bibitem[Eldering et~al., 2017c]{elderingsci}
Eldering, A., Wennberg, P.~O., Crisp, D., Schimel, D.~S., Gunson, M.~R.,
  Chatterjee, A., Liu, J., Schwander, F.~M., Sun, Y., O'Dell, C.~W.,
  Frankenberg, C., Taylor, T., Fisher, B., Osterman, G.~B., Wunch, D.,
  Hakkarainen, J., Tamminen, J., and Weir, B. (2017c).
\newblock {The Orbiting Carbon Observatory-2} early science investigations of
  regional carbon dioxide fluxes.
\newblock {\em Science}, 358(6360).

\bibitem[Grenander, 1950]{grenander1950stochastic}
Grenander, U. (1950).
\newblock Stochastic processes and statistical inference.
\newblock {\em Arkiv f{\"o}r matematik}, 1(3):195--277.

\bibitem[Gromenko et~al., 2012]{gromenko2012estimation}
Gromenko, O., Kokoszka, P., Zhu, L., and Sojka, J. (2012).
\newblock Estimation and testing for spatially indexed curves with application
  to ionospheric and magnetic field trends.
\newblock {\em The Annals of Applied Statistics}, pages 669--696.

\bibitem[Guyon, 1995]{guyon}
Guyon, X. (1995).
\newblock {\em Random fields on a network: modeling, statistics, and
  applications}.
\newblock Springer Science \& Business Media.

\bibitem[Hall and Hosseini-Nasab, 2006]{hall2006properties}
Hall, P. and Hosseini-Nasab, M. (2006).
\newblock On properties of functional principal components analysis.
\newblock {\em Journal of the Royal Statistical Society: Series B (Statistical
  Methodology)}, 68(1):109--126.

\bibitem[Kokoszka and Reimherr, 2019]{kokoszka2019some}
Kokoszka, P. and Reimherr, M. (2019).
\newblock Some recent developments in inference for geostatistical functional
  data.
\newblock {\em Revista Colombiana de Estad{\'\i}stica}, 42(1):101--122.

\bibitem[Kuenzer et~al., 2020]{kuenzer2020principal}
Kuenzer, T., H{\"o}rmann, S., and Kokoszka, P. (2020).
\newblock Principal component analysis of spatially indexed functions.
\newblock {\em Journal of the American Statistical Association}, 00:1--13.

\bibitem[Lahiri, 2003]{lahiri2003central}
Lahiri, S. (2003).
\newblock Central limit theorems for weighted sums of a spatial process under a
  class of stochastic and fixed designs.
\newblock {\em Sankhy{\=a}: The Indian Journal of Statistics}, pages 356--388.

\bibitem[Lahiri et~al., 2002]{lahiri2002asymptotic}
Lahiri, S.~N., Lee, Y., and Cressie, N. (2002).
\newblock On asymptotic distribution and asymptotic efficiency of least squares
  estimators of spatial variogram parameters.
\newblock {\em Journal of Statistical Planning and Inference}, 103(1-2):65--85.

\bibitem[Li et~al., 2019]{li2019spatial}
Li, Y., Huang, C., and H{\"a}rdle, W.~K. (2019).
\newblock Spatial functional principal component analysis with applications to
  brain image data.
\newblock {\em Journal of Multivariate Analysis}, 170:263--274.

\bibitem[Li et~al., 2013]{li2013selecting}
Li, Y., Wang, N., and Carroll, R.~J. (2013).
\newblock Selecting the number of principal components in functional data.
\newblock {\em Journal of the American Statistical Association},
  108(504):1284--1294.

\bibitem[Liu et~al., 2017]{liu2017}
Liu, C., Ray, S., and Hooker, G. (2017).
\newblock Functional principal component analysis of spatially correlated data.
\newblock {\em Statistics and Computing}, 27(6):1639--1654.

\bibitem[Mart{\'\i}nez-Hern{\'a}dez and Genton, 2020]{martinez2020recent}
Mart{\'\i}nez-Hern{\'a}dez, I. and Genton, M.~G. (2020).
\newblock Recent developments in complex and spatially correlated functional
  data.
\newblock {\em arXiv preprint arXiv:2001.01166}.

\bibitem[Menafoglio and Petris, 2016]{menafoglio2016kriging}
Menafoglio, A. and Petris, G. (2016).
\newblock Kriging for hilbert-space valued random fields: The operatorial point
  of view.
\newblock {\em Journal of Multivariate Analysis}, 146:84--94.

\bibitem[O'Dell et~al., 2018]{odell2018}
O'Dell, C.~W., Eldering, A., Wennberg, P.~O., Crisp, D., Gunson, M.~R., Fisher,
  B., Frankenberg, C., Kiel, M., et~al. (2018).
\newblock Improved retrievals of carbon dioxide from the {Orbiting Carbon
  Observatory-2} with the version 8 {ACOS} algorithm.
\newblock {\em Atmospheric Measurement Techniques}, 11:6539--6576.

\bibitem[Osterman et~al., 2018]{v9dug}
Osterman, G., Eldering, A., Avis, C., Chafin, B., O'Dell, C., Frankenberg, C.,
  Fisher, B., Mandrake, L., Wunch, D., Granat, R., and Crisp, D. (2018).
\newblock {\em {Orbiting Carbon Observatory (OCO-2) Level 2 Data Product User's
  Guide, Operational L1 and L2 Data Versions 8 and Lite File Version 9}}.
\newblock Jet Propulsion Laboratory.

\bibitem[Ramsay, 2004]{ramsay}
Ramsay, J.~O. (2004).
\newblock Functional data analysis.
\newblock {\em Encyclopedia of Statistical Sciences}, 4.

\bibitem[Ruggieri et~al., 2018]{ruggieri2018comparing}
Ruggieri, M., Plaia, A., and Di~Salvo, F. (2018).
\newblock Comparing spatial and spatio-temporal fpca to impute large continuous
  gaps in space.
\newblock In {\em Classification,(Big) Data Analysis and Statistical Learning},
  pages 201--208. Springer.

\bibitem[van~der Vaart and Wellner, 1996]{kiessler2009weak}
van~der Vaart, A.~W. and Wellner, J.~A. (1996).
\newblock {\em Weak Conver- gence and Empirical Processes. With Applications to
  Statistics}.
\newblock Springer Sciences+Business Media.

\bibitem[Wang et~al., 2016]{wang2016functional}
Wang, J.-L., Chiou, J.-M., and M{\"u}ller, H.-G. (2016).
\newblock Functional data analysis.
\newblock {\em Annual Review of Statistics and Its Application}, 3:257--295.

\bibitem[Wang et~al., 2019]{wang2019simultaneous}
Wang, Y., Wang, G., Wang, L., and Ogden, R.~T. (2019).
\newblock Simultaneous confidence corridors for mean functions in functional
  data analysis of imaging data.
\newblock {\em Biometrics}.

\bibitem[Yao et~al., 2003]{yao2003}
Yao, F., M{\"u}ller, H.-G., Clifford, A.~J., Dueker, S.~R., Follett, J., Lin,
  Y., Buchholz, B.~A., and Vogel, J.~S. (2003).
\newblock Shrinkage estimation for functional principal component scores with
  application to the population kinetics of plasma folate.
\newblock {\em Biometrics}, 59(3):676--685.

\bibitem[Yao et~al., 2005]{yao2005}
Yao, F., M{\"u}ller, H.-G., and Wang, J.-L. (2005).
\newblock Functional data analysis for sparse longitudinal data.
\newblock {\em Journal of the American Statistical Association},
  100(470):577--590.

\bibitem[Zhang et~al., 2016]{zhang2016}
Zhang, L., Baladandayuthapani, V., Zhu, H., Baggerly, K.~A., Majewski, T.,
  Czerniak, B.~A., and Morris, J.~S. (2016).
\newblock Functional car models for large spatially correlated functional
  datasets.
\newblock {\em Journal of the American Statistical Association},
  111(514):772--786.

\bibitem[Zhao et~al., 2004]{zhao2004functional}
Zhao, X., Marron, J., and Wells, M.~T. (2004).
\newblock The functional data analysis view of longitudinal data.
\newblock {\em Statistica Sinica}, pages 789--808.

\end{thebibliography}
 
 \newpage 
 
 \title{\textbf{Supplementary Material to} \textit{A Geospatial Functional Model For OCO-2 Data with Application on Imputation and Land Fraction Estimation}}
  \author{Xinyue Chang$^1$, Zhengyuan Zhu$^1$, Xiongtao Dai$^1$ \\
   and Jonathan Hobbs$^2$ \\
 \hspace{.2cm}\\
  $^1$Department of Statistics, Iowa State University \\
  $^2$Jet Propulsion Laboratory, California Institute of Technology}
 %\author{}
 %\date{}
 \maketitle
 
 \spacingset{1.75}
 %\section{Appendices}
 \setcounter{page}{1}
 \renewcommand{\thepage}{S.\arabic{page}}
 \renewcommand\thesection{S.\arabic{section}}
 \renewcommand{\theequation}{S.\arabic{equation}}
 \renewcommand{\thefigure}{S.\arabic{figure}}
 \section{Technical Proofs}
 This section provides proofs of asymptotic results for dense functional principal analysis (Theorem \ref{thm1}, \ref{thm2} and \ref{thm3}), and consistent BLUP estimator (Theorem \ref{thm4}) for PC scores. We first recall some notations defined previously. Given a sequences $f(n)$ and $g(n)$, the notation $f(n) = O(g(n))$ means $|f(n)| \leq c_1 |g(n)|$ for some $c_1>0$, and $f(n) = \Omega(g(n))$ means that $|f(n)| \geq c_2 |g(n)|$ for some $c_2 > 0$. Also, $f(n) = \Theta(g(n))$ denotes when both $f(n) = O(g(n))$ and $f(n) = \Omega(g(n))$.
 \subsection{Mean function estimation}
 The mean function model (\ref{mean_model}) can be regarded as 8 separate linear models within each footprint group. Without loss of generality, in the proof we assume covariates be centered, i.e., $\bdone\trans \bdx_p = 0$ and $\BX_p = [\bdone \quad \bdx_p]$ and $\bdbeta_p(w) = \{\beta_0^p(w), \bdbeta_1^p(w)\trans \}\trans $, define total error term as 
 \begin{align*}
 u(w; \bds_i) = \sum_{k=1}^K \xi_k(\bds_i) \phi_k(w) + \epsilon_{i}(w), \quad w \in \mathcal{W},
 \end{align*}
 for $p = 1, \ldots, 8$ and any $w \in \mathcal{W}$,
 \begin{align*}
 \BY_{p}(w) = \BX_p \bdbeta_p(w) + \BU_{p}(w),
 \end{align*}
 where $\BU_{p}(w) = \{u(w; \bds_i)\}_{\bds_i \in \mathcal{S}_{np}}$ which corresponds to $\bdx_p$. So equivalent to (\ref{beta_estimate}), $\wh \bdbeta_p(w) = \left(\BX_p\trans \BX_p\right)^{-1} \BX_p\trans \BY_{p}(w)$, and 
 \begin{align}\label{beta_p}
 \wh \bdbeta_p(w) - \bdbeta_p(w) & = \left(\BX_p\trans \BX_p\right)^{-1} \BX_p\trans \BU_{p}(w), \\ \label{beta_0p}
 \wh \beta_0^p(w) - \beta_0^p(w)& = \dfrac{1}{N_p} \sum_{\bds_i \in \mathcal{S}_{np}} u(w; \bds_i), \\ \label{beta_1p}
 \wh \bdbeta_1^p(w) - \bdbeta_1^p(w) & = (\bdx_p\trans \bdx_p)^{-1} \bdx_p\trans \BU_{p}(w).
 \end{align}
 Specifically, let $\bdepsilon_{p}(w) = \{\epsilon_{i}(w)\}_{\bds_i \in \mathcal{S}_{np}}$ with $\var\{\bdepsilon_{p}(w)\} = \sigma_p^2(w) \BI$ and $\bdxi_{kp} = \{\xi_k(\bds_i)\}_{\bds_i \in \mathcal{S}_{np}}$, then
 \begin{align}  \nonumber
 & \wh \bdbeta_1^p(w) - \bdbeta_1^p(w) \\ \nonumber
 = & (\bdx_p\trans \bdx_p)^{-1}\bdx_p\trans \bdepsilon_{p}(w) + \sum_{k=1}^K (\bdx_p\trans \bdx_p)^{-1} \bdx_p\trans \bdxi_{kp} \phi_k(w) \\ \label{beta1_ueqn}
 := & \Delta_{p}^1(w) + \Delta_{p}^2(w),
 \end{align}
 where $\Delta_{p}^1(w) = (\bdx_p\trans \bdx_p)^{-1}\bdx_p\trans \bdepsilon_{p}(w)$ and $\Delta_{p}^2(w) = \sum_{k=1}^K (\bdx_p\trans \bdx_p)^{-1} \bdx_p\trans \bdxi_{kp} \phi_k(w)$. 
 
 We begin by diagonalizing the matrix $\bdx_p\trans \bdx_p$, writing $\bdx_p\trans \bdx_p = UD U\trans$ where $D$ is diagonal and $U$ is unitary. Define $s_{1,n}^2 = \sum_{\bds_i \in \mathcal{S}_{np}} (u_1\trans \bds_i)^2$ and $s_{2,n}^2 = \sum_{\bds_i \in \mathcal{S}_{np}} (u_2\trans \bds_i)^2$, where $(u_1, u_2)$ is the orthonormal matrix $U$. Hence, $D = U\trans \bdx_p\trans \bdx_p U = \mathrm{diag}(s_{1,n}^2, s_{2,n}^2)$. The centered sampling region of $\mathcal{S}_{np}$ denotes as $\mathcal{R}_{np}$ contains origin and has the same shape as $\lambda_n \mathcal{R}_0$. Define a lattice at stage $n$ with size $z > 0$ as $J_n(z) = \{(\delta_1 h_n i_1, \delta_2 h_n i_2): |i_1| \leq z, |i_2| \leq z\}$, then the largest lattice within $\mathcal{S}_{np}$ has size $z_n = \underset{J_n(z) \subset \mathcal{S}_{np}}{\max} z$. By the fact that it is a scaled sampling region, $z_n^2 = \Theta(N_p) = \Theta(\lambda_n^2/h_n^2)$. Thus, for $d = 1, 2$, 
 \begin{align} \nonumber
 s_{d,n}^2 &= \sum_{\bds_i \in \mathcal{S}_{np}} (u_{d1}L_i + u_{d2}l_i)^2 \\ \nonumber
 & \geq \sum_{|i_1| \leq z_n} \sum_{|i_2| \leq z_n} (u_{d1} \delta_1 h_n i_1 + u_{d2}\delta_2 h_n i_2)^2 \\ \nonumber
 & \geq h_n^2 \{2 z_n u_{d1}^2 \delta_1^2 \sum_{|i_1| \leq z_n} i_1^2 + 2 z_n u_{d2}^2 \delta_2^2 \sum_{|i_2| \leq z_n} i_2^2\} \\ \label{sdn2}
 & = \Omega(h_n^2 z_n^4) = \Omega(\lambda_n^4 / h_n^2)
 \end{align}
 %Similarly, we also have $s_{d,n}^2 = O(\lambda_n^4 / h_n^2)$. 
 Then we introduce the following lemma regarding the weighted sum of error processes.
 \begin{lm} \label{lm:weighted_sum_e}
 	Consider the measurement error process $\epsilon_i(w) = \sigma_p(w)e_i(w)$ at $\bds_i \in \mathcal{S}_{np}$ and $w \in \mathcal{W}$, $e_i(w)$ are i.i.d. mean zero stochastic processes with $\var\{e_i(w)\} = 1$, $w \in W$. In addition,
 	\begin{align*}
 	|e_i(w) - e_i(w')| \leq L_i |w-w'|,
 	\end{align*}
 	where $L_i$ are i.i.d. random variables with $\E L_i^2 < \infty$. If $\E\{\sup_{w\in\mathcal{W}} |e_1(w)|^2\} < \infty$, then 
 	\begin{align*}
 	\sup_{w \in \mathcal{W}} \left|s_{d,n}^{-2} \sum_{\bds_i \in \mathcal{S}_{np}} (u_d\trans \bds_i ) \epsilon_i(w)\right| = O_p(h_n/\lambda_n^2)
 	\end{align*}
 	for $d = 1, 2$.
 \end{lm}
 \begin{proof}
 	For each $n$, let $Z_{ni,d}(w) = s_{d,n}^{-1} (u_d\trans \bds_i) e_i(w)$, $\bds_i \in \mathcal{S}_{np}$, $w \in \mathcal{W}$, be independent stochastic processes and $\rho(w, w') = |w - w'|$. For every $\varepsilon > 0$, let $M_n = \sup_{\bds_i \in \mathcal{S}_{np}} |u_d\trans \bds_i| = O(\lambda_n)$ by mixed-increasing-domain,
 	\begin{align*}
 	& \sum_{\bds_i\in\mathcal{S}_{np}}\E \left[\left\{ \sup_{w \in \mathcal{W}} \left|\dfrac{u_d\trans \bds_i}{s_{d,n}} e_i(w)\right|\right\}^2 I\left\{\sup_{w\in \mathcal{W}} \left|\dfrac{u_d\trans \bds_i}{s_{d,n}} e_i(w)\right| > \varepsilon\right\}\right] \\
 	\leq & \sum_{\bds_i\in\mathcal{S}_{np}}\E \left[\left\{ \sup_{w \in \mathcal{W}} \left|\dfrac{M_n}{s_{d,n}} e_i(w)\right|\right\}^2 I\left\{\sup_{w\in \mathcal{W}} \left|\dfrac{M_n}{s_{d,n}} e_i(w)\right| > \varepsilon\right\}\right] \\
 	= & \dfrac{M_n^2}{s_{d,n}^2} N_p \E \left[\left\{\sup_{w \in \mathcal{W}} |e_i(w)|^2\right\} I \left\{\sup_{w \in \mathcal{W}} |e_i(w)| > \varepsilon \dfrac{s_{d,n}}{M_n}\right\}\right] \\
 	= &  o(1)
 	\end{align*}
 	by $N_p (M_n/ s_{d,n})^2 = O(1)$ and DCT with $\E\{\sup_{w\in\mathcal{W}} |e_i(w)|^2\} < \infty$. In addition, for $w, w' \in \mathcal{W}$, 
 	\begin{align*}
 	&\sum_{\bds_i \in \mathcal{S}_{np}} \cov \{Z_{ni,d}(w), Z_{ni,d}(w')\} \\
 	=& s_{d,n}^{-2} \sum_{\bds_i \in \mathcal{S}_{np}} (u_d\trans \bds_i)^2 \cov\{e_i(w), e_i(w')\} \\
 	=&  \cov\{e_1(w), e_1(w')\} \leq \sqrt{\E\{e_1(w)\}^2 \E\{e_1(w')\}^2} < \infty
 	\end{align*}
 	by assumptions on process $e_i(w)$. Thus the previous two displays give marginal weak convergence to a Gaussian limit by Lindeberg condition theorem. For every sequence $\delta_n \downarrow 0$, by the Lipschitz continuity assumption, 
 	\begin{align*}
 	& \sup_{\rho(w,w') < \delta_n} \sum_{\bds_i \in \mathcal{S}_{np}}\E\{Z_{ni,d}(w) - Z_{ni,d}(w')\}^2 \\
 	\leq & \delta_n^2 \sum_{\bds_i \in \mathcal{S}_{np}} \dfrac{(u_d\trans \bds_i)^2}{s_{d,n}^2} \E L_i^2 = O(\delta_n^2) \rightarrow 0. 
 	\end{align*}
 	Define a random semimetric (Section~2.11.1 in \cite{kiessler2009weak}) by 
 	\begin{align*}
 	d_n^2(w, w') = \sum_{\bds_i \in \mathcal{S}_{np}} \{Z_{ni,d}(w) - Z_{ni,d}(w')\}^2,
 	\end{align*}
 	which is upper bounded as
 	\begin{align*}
 	d_n^2(w, w') &\leq \sum_{\bds_i \in \mathcal{S}_{np}} \left(\dfrac{u_d\trans \bds_i}{s_{d,n}}\right)^2 L_i^2 (w- w')^2 \\
 	& \leq B_n^2 \rho^2(w, w'),
 	\end{align*}
 	where $B_n^2 = \sum_{\bds_i \in \mathcal{S}_{np}} \left(\dfrac{u_d\trans \bds_i}{s_{d,n}}\right)^2 L_i^2$. Let $N(\varepsilon, \mathcal{F}, ||\cdot||)$ be the minimal number of balls $\{g: ||g-f|| < \varepsilon\}$ of radius $\varepsilon$ needed to cover the set $\mathcal{F}$. For any sequence $\delta_n \downarrow 0$,
 	\begin{align*}
 	\int_0^{\delta_n} \sqrt{\log N(\varepsilon, \mathcal{W}, d_n)}d\varepsilon \leq & \int_0^{\delta_n} \sqrt{\log N(\varepsilon, \mathcal{W}, B_n\rho)}d\varepsilon \\
 	= & B_n \int_0^{\delta_n/B_n} \sqrt{\log N(u, \mathcal{W}, \rho)}du = O_p(\delta_n^{1/2}).
 	\end{align*}
 	Since $N(u, \mathcal{W}, \rho) \leq \lceil |\mathcal{W}|/(2u) \rceil$ and $B_n = O_p(1)$ by assumption \ref{cond:mm}. By Theorem~2.11.1 in \cite{kiessler2009weak}, $\sup_{w\in \mathcal{W}} \left|\sum_{\bds_i \in \mathcal{S}_{np}} Z_{ni}(w)\right|$ converges to a tight limit weakly. The proof is complete by noticing that $s_{d,n}^{-1} \sigma_p(w)Z_{ni}(w) = s_{d,n}^{-2} (u_d\trans \bds_i) \epsilon_i(w)$ and $s_{d,n} = O(\lambda_n^2 /h_n)$.
 \end{proof}
 
 Applying Lemma \ref{lm:weighted_sum_e}, we have
 \begin{align*}
 \sup_{w \in \mathcal{W}}||\Delta_p^1(w)||_2  = & \sup_{w \in \mathcal{W}} ||(U\trans \bdx_p\trans \bdx_p U)^{-1} U \trans \bdx_p\trans \bdepsilon_p(w)||_2 \\
 \leq & \sup_{w \in \mathcal{W}} ||\{s_{1,n}^{-2} \sum_{\bds_i \in \mathcal{S}_{np}} (u_1\trans \bds_i) \epsilon_{i}(w), s_{2,n}^{-2} \sum_{\bds_i \in \mathcal{S}_{np}} (u_2\trans \bds_i) \epsilon_{i}(w)\}\trans ||_2 \\
 = & O_p(h_n/\lambda_n^2)
 %= & O_p(1/\sqrt{s_{d,n}^2}) = O_p(h_n/\lambda_n^2).
 \end{align*}
 by condition \ref{cond:mm}. 
 Similarly, for the second term $\Delta_p^2(w)$, 
 \begin{align*}
 ||(\bdx_p\trans \bdx_p)^{-1} \bdx_p\trans \bdxi_{kp}||_2 = ||\{s_{1,n}^{-2} \sum_{\bds_i \in \mathcal{S}_{np}} (u_1\trans \bds_i) \xi_k(\bds_i), s_{2,n}^{-2} \sum_{\bds_i \in \mathcal{S}_{np}} (u_2\trans \bds_i) \xi_k(\bds_i)\}\trans ||_2.
 \end{align*}
 Since $\bdx_p$ is centered, the functions $u_1\trans \bds_i$ and $u_2 \trans \bds_i$ satisfy the (F.1) condition of \cite{lahiri2003central} with limiting function equal to 1. Let $\gamma_{d,n}^2 = M_{d,n}^2 \lambda_n^2 h_n^{-2} s_{d,n}^{-2}$ where $d = 1,2$, $M_{d,n} = \sup\{|u_d\trans \bds|: \bds \in \mathcal{R}_{np} \} = O(\lambda_n)$. Together with result (\ref{sdn2}), it implies that $\gamma_{d,n}^2 = O(1)$. Under conditions \ref{cond:alpha}--\ref{cond:g} and \ref{cond:xi_k}, by Proposition 4.1 and Theorem 4.1 in \cite{lahiri2003central},
 \begin{align*}
 s_{d,n}^{-2} \sum_{\bds_i \in \mathcal{S}_{np}} (u_d\trans \bds_i) \xi_k(\bds_i) = O_p(h_n^{-1} (s_{d,n}^2)^{-1/2}) = O_p(1/\lambda_n^2),
 \end{align*}
 where $d = 1, 2$. It follows that 
 \begin{align*}
 \sup_{w \in \mathcal{W}} ||\Delta_p^2(w)||_2 \leq \sum_{k=1}^K ||(\bdx_p\trans \bdx_p)^{-1} \bdx_p\trans \bdxi_{kp}||_2 \sup_{w \in \mathcal{W}} |\phi_k(w)| = O_p(1/\lambda_n^2).
 \end{align*}
 Combing two terms $\Delta_p^1(w)$ and $\Delta_p^2(w)$, we have
 \begin{align} \label{beta1_est_conv_rate}
 \sup_{w \in \mathcal{W}} ||\wh \bdbeta_1^p(w) - \bdbeta_1^p(w)||_2 = O_p(1/\lambda_n^2) 
 \end{align}
 
 For the intercept estimate $\wh \beta_0^p(w)$, following eqn. (\ref{beta_0p}), 
 \begin{align*}
 \sup_{w \in \mathcal{W}}||\wh \beta_0^p(w) - \beta_0^p(w)||_2 \leq & \sum_{k = 1}^K \sup_{w \in \mathcal{W}}|\phi_k(w)| \left\lvert \dfrac{1}{N_p} \sum_{\bds_i \in \mathcal{S}_{np}} \xi_k(\bds_i)\right\rvert \\
 & + \sup_{w \in \mathcal{W}} \left|N_p^{-1} \sum_{\bds_i \in \mathcal{S}_{np}} \epsilon_i(w)\right|.
 \end{align*}
 Similarly by the Theorem 4.1 in \cite{lahiri2003central}, under mixing conditions assumed, $\dfrac{1}{N_p}\sum_{\bds_i \in \mathcal{S}_p} \\ \xi_k(\bds_i) = O_p(N_p^{-1/2}h_n^{-1})$. As a simpler case than Lemma \ref{lm:weighted_sum_e}, $\sup_{w \in \mathcal{W}} |N_p^{-1} \sum_{\bds_i \in \mathcal{S}_{np}} \epsilon_i(w)| = O_p(N_p^{-1/2})$ by the same argument. Therefore, 
 $$\sup_{w \in \mathcal{W}} ||\wh \beta_0^p(w) - \beta_0^p(w)||_2 = O_p(1/\lambda_n),$$
 and furthermore with the previous result (\ref{beta1_est_conv_rate}),  
 \begin{align} \label{beta_est_order}
 \sup_{w \in \mathcal{W}} ||\wh \bdbeta(w) - \bdbeta(w)||_2 = O_p(1/\lambda_n).
 \end{align}
 For a given $\bds_i \in \mathcal{S}$, recall $p_i = q(\bds_i)$, 
 \begin{align} \nonumber
 & \sup_{w \in \mathcal{W}} |\wh \mu(w; \bds_i) - \mu(w; \bds_i)| \\ \nonumber
 \leq & \sup_{w \in \mathcal{W}} |\wh \beta_0^{p_i}(w) - \beta_0^{p_i}(w)| + \sup_{w \in \mathcal{W}} |\bds_i \trans \{\wh \bdbeta_1^{p_i}(w) - \bdbeta_1^{p_i}(w)\}| \\ \label{mean_est_order}
 = & O_p(1/\lambda_n)
 \end{align}
 As shown in (\ref{beta_est_order}) and (\ref{mean_est_order}), the proof of Theorem \ref{thm1} is done.
 
 \subsection{Covariance function estimation}
 Based on condition \ref{cond:mm}, $\sup_{w \in \mathcal{W}} |\sum_{k=1}^K \xi_k(\bds_i) \phi_k(w) + \epsilon_{i}(w)| = O_p(1)$. For any $w, w' \in \mathcal{W}$, by Theorem \ref{thm1},
 \begin{align*}
 & \sup_{w,w'\in \mathcal{W}}|\wh R_f(w, w') - R_f(w, w')| \\
 \leq & \sup_{w, w'\in\mathcal{W}}|\mathcal{A}_{1,n} + \mathcal{A}_{2,n} + \mathcal{A}_{3,n}| + O_p(1/\lambda_n)
 \end{align*}
 where 
 \begin{align*}
 \mathcal{A}_{1, n}& = \dfrac{1}{N-1}\sum_{\bds_i \in \mathcal{S}_n} \left\{\sum_{k=1}^K \xi_k(\bds_i) \phi_k(w)\right\}\epsilon_{i}(w') + \dfrac{1}{N-1}\sum_{\bds_i \in \mathcal{S}_n} \left\{\sum_{k=1}^K \xi_k(\bds_i) \phi_k(w')\right\} \epsilon_{i}(w) \\
 \mathcal{A}_{2,n}& = \dfrac{1}{N-1}\sum_{\bds_ \in \mathcal{S}_n} \left\{\sum_{k=1}^K \xi_k(\bds_i) \phi_k(w)\right\}\left\{\sum_{k=1}^K \xi_k(\bds_i) \phi_k(w')\right\} - \sum_{k=1}^K \lambda_k \phi_k(w) \phi_k(w')\\
 \mathcal{A}_{3,n} & = \dfrac{1}{N-1} \sum_{\bds_i \in \mathcal{S}_n} \epsilon_{i}(w) \epsilon_{i}(w') - \dfrac{1}{N-1}\sum_{p=1}^8 N_p \wh R_{\epsilon, p}(w, w')
 \end{align*}
 For the first term in $\mathcal{A}_{1,n}$,
 \begin{align*}
 & \sup_{w, w'\in \mathcal{W}}\left \lvert\dfrac{1}{N-1}\sum_{\bds_i \in \mathcal{S}_n} \left\{\sum_{k=1}^K \xi_k(\bds_i) \phi_k(w)\right\}\epsilon_{i}(w')\right\rvert \\
 %= & \sup_{w, w'\in \mathcal{W}} \left\lvert \sum_{k=1}^K \phi_k(w)\dfrac{1}{N-1} \sum_{\bds_i \in \mathcal{S}_n} \xi_k(\bds_i)\epsilon_{i}(w')\right\rvert \\
 \leq & \sum_{k=1}^K \sup_{w\in\mathcal{W}}|\phi_k(w)| \dfrac{1}{N-1} \sum_{p=1}^8 \sup_{w' \in\mathcal{W}} \left\lvert \sum_{\bds_i \in \mathcal{S}_{np}} \xi_k(\bds_i)\epsilon_{i}(w')\right\rvert.
 \end{align*}
 Conditioning on $\xi_k(\bds_i)$, because $\epsilon_i(w)$ and $\xi_k(\bds_i)$ are independent, then $\sum_{\bds_i \in \mathcal{S}_{np}} \xi_k(\bds_i)\epsilon_{i}(w')$ is a weighted sum of the independent error processes. Moreover, by condition \ref{cond:xi_k}, we have 
 \[
 \sup_{w' \in \mathcal{W}}\left|\sum_{\bds_i \in \mathcal{S}_{np}} \xi_k(\bds_i)\epsilon_{i}(w')\right| = O_p\left(\sqrt{\sum_{\bds_i \in \mathcal{S}_{np}} \xi_k^2(\bds_i)}\right) = O_p(\sqrt{N_p}), 
 \]
 by using the same proof in Lemma \ref{lm:weighted_sum_e}. 
 Therefore, the converge rate is
 \begin{align} \label{A1n_conv_rate}
 \sup_{w,w'\in\mathcal{W}} |\mathcal{A}_{1,n}| = O_p(h_n/\lambda_n).
 \end{align}
 
 For the second term $\mathcal{A}_{2,n}$, 
 \begin{align*} \nonumber
 \mathcal{A}_{2,n} = & \dfrac{1}{N-1}\sum_{i=1}^N \left\{\sum_{k_1 =1}^K \sum_{k_2 = 1}^K \xi_{k_1}(\bds_i) \xi_{k_2}(\bds_i) \phi_{k_1}(w) \phi_{k_2} (w')\right\} \\ \nonumber
 & -\sum_{k=1}^K \lambda_k \phi_k(w) \phi_k(w') \\ 
 = & \sum_{k=1}^K \left\{\dfrac{1}{N-1}\sum_{i=1}^N \xi_k^2(\bds_i) - \lambda_k\right\} \phi_k(w) \phi_k(w') \\ 
 & + \sum_{k_1 = 1}^K \sum_{k_2 \neq k_1}^K \phi_{k_1}(w) \phi_{k_2}(w') \dfrac{1}{N-1}\sum_{i=1}^N \xi_{k_1}(\bds_i) \xi_{k_2}(\bds_i)
 \end{align*}
 By the definition of strong mixing coefficient (\ref{mix_coef}), the mixing coefficient of $\xi_k^2(\cdot)$ is less or equal to $\alpha_k(a; b)$. Thus, we also have conditions \ref{cond:alpha}--\ref{cond:g} satisfied, by Proposition 4.1 and Theorem 4.1 in \cite{lahiri2003central}, $\dfrac{1}{N-1}\sum_{i=1}^N \xi_k^2(\bds_i) - \lambda_k = O_p(1/\lambda_n)$. Similarly, the mixing coefficient for the randome field $\{\xi_{k1}(\bds) \xi_{k2}(\bds): \bds \in \mathbb{R}^2\}$ is less or equal to $\alpha_{k1}(a; b) + \alpha_{k2}(a; b)$ because of independence \citep{bradley2005basic}. Then we have $\dfrac{1}{N-1}\sum_{i=1}^N \xi_{k_1}(\bds_i) \xi_{k_2}(\bds_i) = O_p(1/\lambda_n)$ as well. We can have an uniform bound of $|\mathcal{A}_{2,n}|$ as
 \begin{align*}
 |\mathcal{A}_{2,n}| \leq & \sum_{k=1}^K \left\lvert \dfrac{1}{N-1}\sum_{i=1}^N \xi_k^2(\bds_i) - \lambda_k\right\rvert \sup_{w \in \mathcal{W}}|\phi_k(w)| \sup_{w' \in \mathcal{W}} |\phi_k(w')| \\
 & + \sum_{k_1 = 1}^K \sum_{k_2 \neq k_1}^K \sup_{w \in \mathcal{W}} |\phi_{k_1}(w)| \sup_{w' \in \mathcal{W}}|\phi_{k_2}(w')| \left\lvert\dfrac{1}{N-1}\sum_{i=1}^N \xi_{k_1}(\bds_i) \xi_{k_2}(\bds_i)\right\rvert,
 \end{align*}
 which implies 
 \begin{align} \label{A2n_conv_rate}
 \sup_{w,w'\in \mathcal{W}} |\mathcal{A}_{2,n}| = O_p(1/\lambda_n).
 \end{align}
 
 Considering the third term $\mathcal{A}_{3,n}$, we defined $\Delta_p(w; \bds_i) = r(w; \bds_{i_1}) - 2r(w; \bds_i) + r(w; \bds_{i_2})$ where $\bds_{i_1}$ and $\bds_{i_2}$ are defined as nearest locations of $\bds_i$ in positive and negative latitude directions respectively, 
 \begin{align*}
 & \sup_{w,w' \in \mathcal{W}} |\wh R_{\epsilon,p}(w, w') - R_{\epsilon, p}(w, w)| \\
 = & \sup_{w,w' \in \mathcal{W}} \left\lvert\dfrac{1}{6\wt N_p} \sum_{\bds_i \in \mathcal{S}_{np}} \Delta_p(w; \bds_i)\Delta_p(w'; \bds_i) - R_{\epsilon, p}(w, w)\right\rvert.
 \end{align*}
 Based on our spatial asymptotic framework,  condition \ref{cond:xi_k_smooth} and \ref{cond:mm},
 \begin{align} \nonumber
 & \sup_{w\in\mathcal{W}} |f(w; \bds_{i_1}) - 2f(w; \bds_i) + f(w; \bds_{i_2})| \\ \nonumber
 \leq & \sup_{w\in\mathcal{W}}||\bdbeta_1^p(w)||_2||\bds_{i_1} - 2\bds_{i} + \bds_{i_2}||_2 \\ \nonumber
 & + \sum_{k=1}^K \sup_{w \in \mathcal{W}} |\phi_k(w)| |\xi_k(\bds_{i_1}) - 2\xi_k(\bds_i) + \xi_k(\bds_{i_2})| \\ \label{f_2nd_order}
 =& O_p(h_n^{\beta_1}), \quad \beta_1 = \min\{1, \beta\}.
 \end{align}
 Also, we notice that $e_i(w) e_i(w')$ satisfy the Lipschitz continuity, 
 \begin{align*}
 & |e_i(w_1)e_i(w_2) - e_i(w_1')e_i(w_2')| \nonumber \\
 \leq & \sup_{w \in \mathcal{W}} |e_i(w)| L_i \{|w_1 - w_1'| + |w_2 - w_2'|\}, 
 \end{align*}
 where $\E\{\sup_{w \in \mathcal{W}} |e_i(w)| L_i\}^2 < \infty$ by condition \ref{cond:mm}. Extending the proof of Lemma \ref{lm:weighted_sum_e} to 2-dimensional case, then it is easy to see that 
 \begin{align} \label{Re_conv_rate}
 \sup_{w,w'\in \mathcal{W}} |\dfrac{1}{N_p} \sum_{\bds_i \in S_{np}} e_i(w)e_i(w') - R_e(w, w')| = O_p(N_p^{-1/2}).
 \end{align}
 %By the definition of $\bds_{i_1}$ and $\bds_{i_2}$, they are uniquely belong to one location $\bds_i$, 
 By previous two results (\ref{f_2nd_order}) and (\ref{Re_conv_rate}), and $\wt N_p$ has the same order as $N_p$ based on the asymptotic framework, 
 \begin{align*}
 & \sup_{w,w' \in \mathcal{W}} |\wh R_{\epsilon,p}(w, w') - R_{\epsilon, p}(w, w)| \\
 \leq &  \sup_{w, w' \in\mathcal{W}} \left\lvert\dfrac{1}{6 \wt N_p} \sum_{\bds_i \in \mathcal{S}_{np}} \{\epsilon_{i_1}(w) - 2\epsilon_{i}(w) + \epsilon_{i_2}(w)\}\{\epsilon_{i_1}(w') - 2\epsilon_{i}(w') + \epsilon_{i_2}(w')\} \right. \\
 & \left. - \sigma_p(w) \sigma_p(w') R_e(w,w') \right\rvert + O_p(h_n^{\beta_1}) \\
 \leq &  \sup_{w, w' \in\mathcal{W}} \left\lvert\dfrac{1}{6\wt N_p} \sum_{\bds_i \in \mathcal{S}_{np}} \{-2\epsilon_{i_1}(w)\epsilon_{i}(w') + \epsilon_{i_1}(w)\epsilon_{i_2}(w') - 2\epsilon_i(w) \epsilon_{i_1}(w') \right. \\
 & \left.  - 2\epsilon_i(w) \epsilon_{i_2}(w') + \epsilon_{i_2}(w) \epsilon_{i_1}(w') - 2\epsilon_{i_2}(w) \epsilon_i(w')\} \right\rvert + O_p(h_n/\lambda_n + h_n^{\beta_1}).
 \end{align*}
 The spatial asymptotic framework is based on lattice, then $\sum_{\bds_i \in \mathcal{S}_{np}} \epsilon_{i_1}(w) \epsilon_i(w')$ can be written as sum of two sequences taking independent processes by choosing $\bds_i$ every other latitude. This applies to the other products above as well and conditions for Lemma \ref{lm:weighted_sum_e} are satisfied. Thus, 
 \begin{align} \label{Reps_conv_rate}
 \sup_{w,w' \in \mathcal{W}} |\wh R_{\epsilon,p}(w, w') - R_{\epsilon, p}(w, w)| = O_p(h_n^{\beta_1} + h_n/\lambda_n)
 \end{align}
 where $\beta_1 = \min\{1, \beta\}$. With result (\ref{Re_conv_rate}), it follows that
 \begin{align*}
 & \sup_{w\in\mathcal{W}} \left\lvert \dfrac{1}{N-1}\sum_{i=1}^N \epsilon_{i}(w)\epsilon_{i}(w') - \dfrac{1}{N-1} \sum_{p=1}^8 N_p \wh R_{\epsilon,p}(w, w') \right\rvert\\
 %\leq & \dfrac{1}{N-1} \sum_{p=1}^8 \sup_{w\in\mathcal{W}} \left\lvert\sum_{\bds_i \in \mathcal{S}_{np}}\epsilon_{i,w}^2 - N_p\sigma_p^2(w) \right\rvert + \dfrac{1}{N-1} \sum_{p=1}^8 \sup_{w\in\mathcal{W}} |N_p \wh \sigma_p^2(w) - N_p\sigma^2_p(w)|\\
 \leq & \dfrac{1}{N-1} \sum_{p=1}^8 \left\{\sup_{w\in \mathcal{W}} \left|\sum_{\bds_i \in \mathcal{S}_{np}} \epsilon_{i}(w)\epsilon_{i}(w') - N_p R_{\epsilon,p}(w, w') \right| \right. \\
 & \left. + \sup_{w\in\mathcal{W}} |N_p \wh R_{\epsilon,p}(w, w')  - N_p R_{\epsilon,p}(w, w') |\right\} \\
 = & O_p(h_n^{\beta_1} + h_n/\lambda_n).
 \end{align*}
 Now we verify that $\mathcal{A}_{3,n} = O_p(h_n/\lambda_n + h_n^{\beta_1})$. Finally, together with (\ref{A1n_conv_rate}) and (\ref{A2n_conv_rate}), the proof completes.
 
 \subsection{Principal component analysis}
 Similar to result (S.2) in \cite{li2013selecting}, using (2.8) in \cite{hall2006properties} with $K < \infty$, we have the asymptotic expansion as 
 \begin{align*}
 \wh \phi_k(w) - \phi_k(w) = & \left\{ \sum_{\substack{k' = 1 \\ k' \neq k}}^K \dfrac{\lambda_{k'} \phi_{k'}(w)}{(\lambda_k - \lambda_{k'}) \lambda_k} \int \int (\wh R_f - R_f) \phi_k \phi_{k'} \right.\\
 & \left. - \lambda_k^{-1}\phi_k(w) \int \int (\wh R_f - R_f)\phi_k\phi_k \right. \\
 & \left. + \lambda_{k}^{-1} \int (\wh R_f - R_f)(x,w) \phi_k(x)dx \right\} \{1+o_p(1)\}.
 \end{align*}
 By Theorem \ref{thm2}, we have
 \begin{align} \label{phik_order}
 \sup_{w \in \mathcal{W}} |\wh \phi_k(w) - \phi_k(w)| = O_p(1/\lambda_n + h_n^{\beta_1}).
 \end{align}
 Since we estimate scores $u_k(\bds_i)$ by (\ref{xik_estimate}),
 \begin{align*}
 & |\wh u_k(\bds_i) - u_k(\bds_i)| \\
 \leq & \int |r(w; \bds_i) - \wh \mu(w; \bds_i)| |\wh \phi_k(w) - \phi_k(w) | dw + \int |\mu(w; \bds_i) - \wh \mu(w, \bds_i)| |\phi_k(w)| dw \\
 = & O_p(1/\lambda_n + h_n^{\beta_1}).
 \end{align*}
 Applying triangle inequality, and result (\ref{Reps_conv_rate})--(\ref{phik_order}), 
 \begin{align*}
 & |\wh \tau_k(p) - \tau_k(p)| \\
 = & |\int \wh R_{\epsilon, p}(w, w') \wh \phi_k(w) \wh \phi_k(w')dw dw' - \int R_{\epsilon, p}(w, w') \phi_k(w) \phi_k(w') dwdw'| \\
 = & O_p(1/\lambda_n + h_n^{\beta_1}), 
 \end{align*}
 for all $p = 1, \ldots, 8$.
 
 \subsection{BLUP for principal component scores}
 For a distance lag $h \in \mathbb{R}$, let $\wt \gamma_k(h)$ be the sample semivariogram (\ref{sample_variogram}) calculated using $u_k(\bds_i)$,
 \begin{align*}
 & \wt \gamma_k(h) - \gamma_k(h; \theta_k) - \dfrac{1}{2N(h)}\sum_{d(\bds_i, \bds_j)=h} \{\tau_k(p_i) + \tau_k(p_j)\}\\
 = & \dfrac{1}{2N(h)}\sum_{d(\bds_i, \bds_j)=h} \{\xi_k(\bds_i) - \xi_k(\bds_j)\}^2 - \gamma_k(h; \theta_k) \\ 
 & + \dfrac{1}{N(h)}\sum_{d(\bds_i, \bds_j)=h} \{\xi_k(\bds_i) - \xi_k(\bds_j)\}(e_{ik} - e_{jk}) \\
 & + \dfrac{1}{2N(h)}\sum_{d(\bds_i, \bds_j)=h} \{(e_{ik} - e_{jk})^2 - \tau_k(p_i) - \tau_k(p_j)\}.
 \end{align*}
 Let $\bdh \in \mathbb{R}^2$ be a vector that have length of $h$, by conditions \ref{cond:alpha}--\ref{cond:g} and \ref{cond:xi_k}, Proposition~4.1 and Theorem~4.1 in \cite{lahiri2003central}, 
 \begin{align*}
 \dfrac{1}{N} \sum_{\bds_i \in \mathcal{S}_n} [\{\xi_k(\bds_i) - \xi_k(\bds_i + \bdh)\}^2 - 2 \gamma_k(h; \theta_k)] = O_p(N^{-1/2}h_n^{-1}).
 \end{align*}
 Following proof of Theorem~3.3 in \cite{lahiri2002asymptotic}, we obtain
 \begin{align*}
 \dfrac{1}{2N(\bdh)}\sum_{\bds_i-\bds_j=\bdh} \{\xi_k(\bds_i) - \xi_k(\bds_j)\}^2 - \gamma_k(h; \theta_k) = o_p(1)
 \end{align*}
 for any $\bdh$ with great circle length of $h$. Thus,
 \begin{align} \label{xi_k2_conv}
 \dfrac{1}{2N(h)}\sum_{d(\bds_i, \bds_j)=h} \{\xi_k(\bds_i) - \xi_k(\bds_j)\}^2 - \gamma_k(h; \theta_k) = o_p(1).
 \end{align}
 Using the similar argument, we have consistency for the other two terms as well. Since $e_{ik}$ and $\xi_k(\bds_i)$ are independent, then 
 \begin{align} \label{xi_ke_conv}
 \dfrac{1}{N(h)} \sum_{d(\bds_i - \bds_j)=h} \{\xi_k(\bds_i) - \xi_k(\bds_j)\}(e_{ik} - e_{jk}) = o_p(1),
 \end{align}
 by using condition \ref{cond:alpha}--\ref{cond:mm} and Theorem~4.2 in \cite{lahiri2003central}. Also, 
 \begin{align} \label{e_k2_conv}
 \dfrac{1}{2N(h)} \sum_{d(\bds_i, \bds_j)=h} \{(e_{ik} - e_{jk})^2 - \tau_k(p_i) - \tau_k(p_j)\} = o_p(1), 
 \end{align}
 by condition \ref{cond:mm} and weak law of large numbers. Finally, 
 by (\ref{xi_k2_conv})--(\ref{e_k2_conv}) and Theorem \ref{thm3},  
 \begin{align*}
 & |\wh \gamma_k(h) - \gamma_k(h; \theta_k)|  \\
 = & |\wh \gamma_k(h) - \wt \gamma_k(h) + \wt \gamma_k(h) - \gamma_k(h; \theta_k)| \\
 \leq  & \dfrac{1}{2N(h)} \sum_{d(\bds_i, \bds_j)=h} |\{\wh u_k(\bds_i) - \wh u_k(\bds_j)\}^2 - \{u_k(\bds_i) - u_k(\bds_j)\}^2| \\
 & + \dfrac{1}{2N(h)}\sum_{d(\bds_i, \bds_j)=h} |\tau_k(p_i) + \tau_k(p_j) - \wh \tau_k(p_i) - \wh \tau_k(p_j)| + o_p(1) \\
 = & o_p(1)
 \end{align*}
 Theorem~3.1 in \cite{lahiri2002asymptotic} also hold if distance lag $h_l, l = 1, \ldots, L$ is defined as isotropic distance, then by conditions \ref{c1}--\ref{c3} and the previous display, we have $\wh \theta_k - \theta_k = o_p(1)$.
 
 \section{Additional Application Results} \label{sec:appendix:application}
 
 In Fig \ref{fig:rmse1} and \ref{fig:rmse2}, the 8 heat maps represent average RRMSE (\ref{rmse_eqn}) in 128 imputations for removed region $T_r(\cdot)$ from $r = 1$ to $r = 8$. 
 \begin{figure}[ht]
 	\begin{subfigure}{1\textwidth}
 		\centering
 		\includegraphics[trim = 0cm 6cm 0cm 4cm, clip = true, width = 0.9\textwidth]{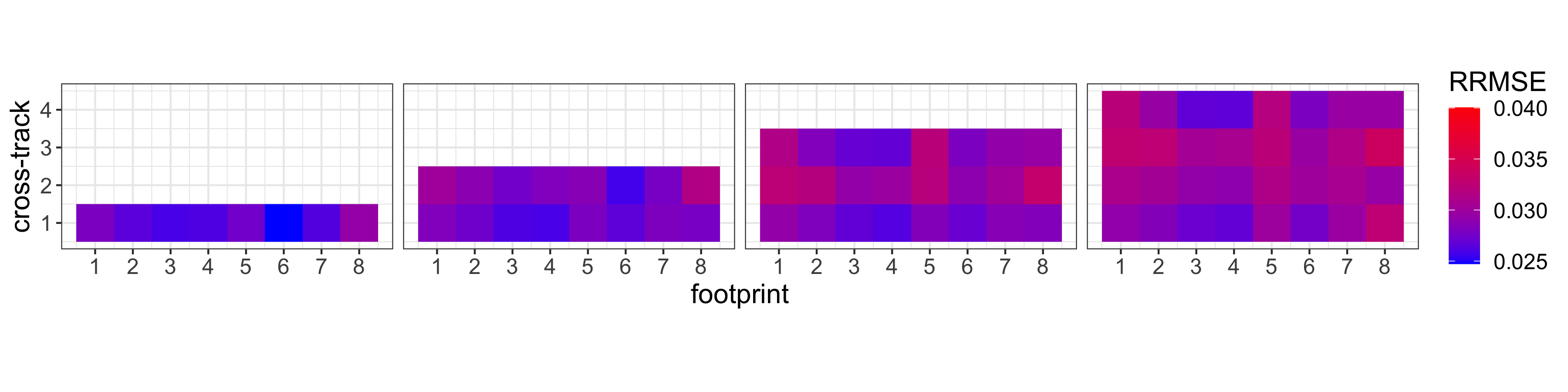}
 		\caption{}
 		\label{fig:rmse1}
 	\end{subfigure}
 	
 	\begin{subfigure}{1\textwidth}
 		\centering
 		\includegraphics[trim = 4cm 0cm 4cm 0cm, clip = true, width = 0.9\textwidth]{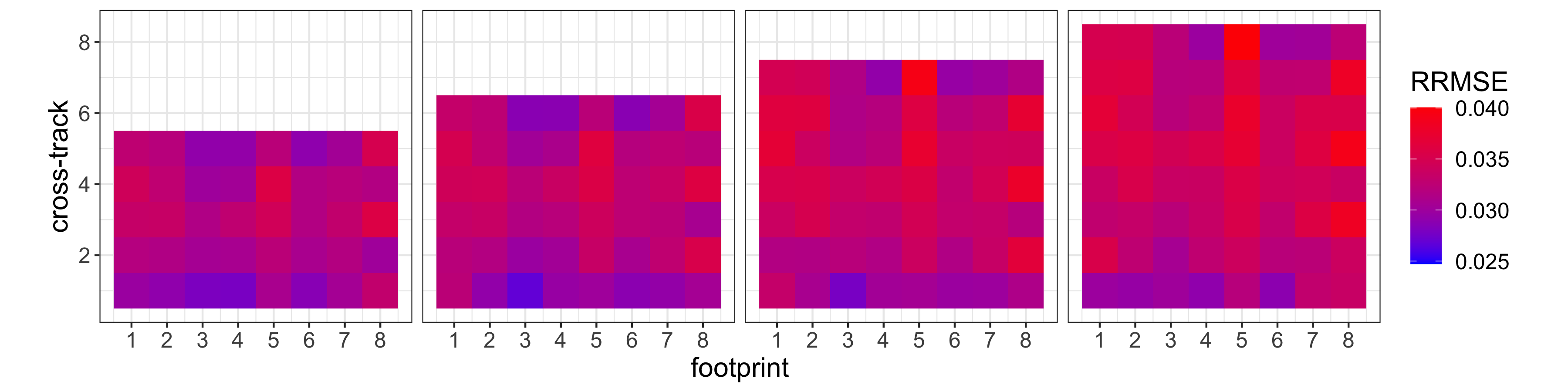}
 		\caption{}
 		\label{fig:rmse2}
 	\end{subfigure}
 	\caption{Removed cross-tracks colored by RRMSE in radiance imputation: (a) average RRMSE in all implementations for 1-4 cross-tracks removed; (b) average RRMSE in all implementations for 5-8 cross-tracks removed.}
 \end{figure}
 Overall the imputation is very close to the observed value, with the root relative MSE less than 0.04. The imputation performance deteriorates as the size of the missing region increases, which is as expected, as the dependence decreases with distance, and it is more difficult to fill in a larger gap. The heat maps also show some variation across the footprints in imputation accuracy: locations in boundaries of the layout or middle are harder to impute than others.
 
 To evaluate how well our ordinary kriging predictor for principal component score performs in radiance imputation, we define the RMSPE (Root Mean Squared Prediction Error) as
 \begin{align}
 \label{rpmse_eqn}
 \text{RMSPE} = \sqrt{\dfrac{1}{m} \sum_{w \in \mathcal{W}_a} \left[\sum_{k=1}^K \{u_k(\bds_0) - \wh \xi_k(\bds_0)\} \wh \phi_k (w)\right]^2}.
 \end{align}
 \begin{figure}[H]
 	\centering
 	\includegraphics[width = 0.8\textwidth]{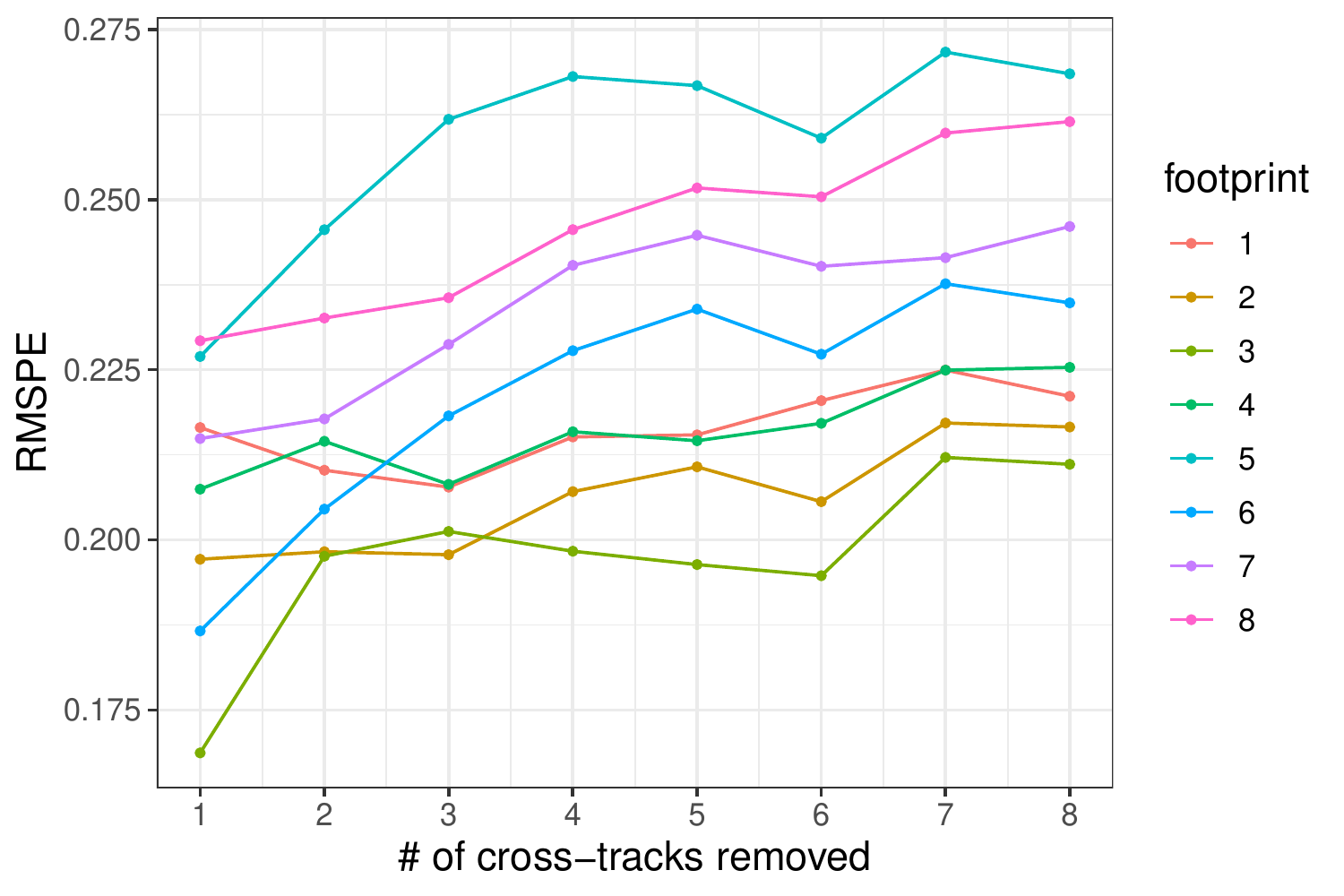}
 	\caption{Average RMSPE in all implementations for different footprints, against 1--8 number of cross-tracks removed.}
 	\label{fig:rpmse_plot}
 \end{figure}
 \noindent Fig \ref{fig:rpmse_plot} shows the consistent results to Fig \ref{fig:rmse_rpmse_plot}, and footprint 5 is the worst in terms of predicting component scores.
 Similarly, predicted mean square error results are summarized in Fig \ref{fig:rpmse1} and \ref{fig:rpmse2}. It is clear that prediction becomes worse as points are closer to the center and further to the edge of imputed region. 
 %Fig \ref{fig:rpmse_plot} shows a similar pattern that as removed region becomes larger, prediction error is higher for the cross-track of center point.
 \begin{figure}[ht]
 	\begin{subfigure}{1\textwidth}
 		\centering
 		\includegraphics[trim = 0cm 6cm 0cm 4cm, clip = true, width = 0.95\textwidth]{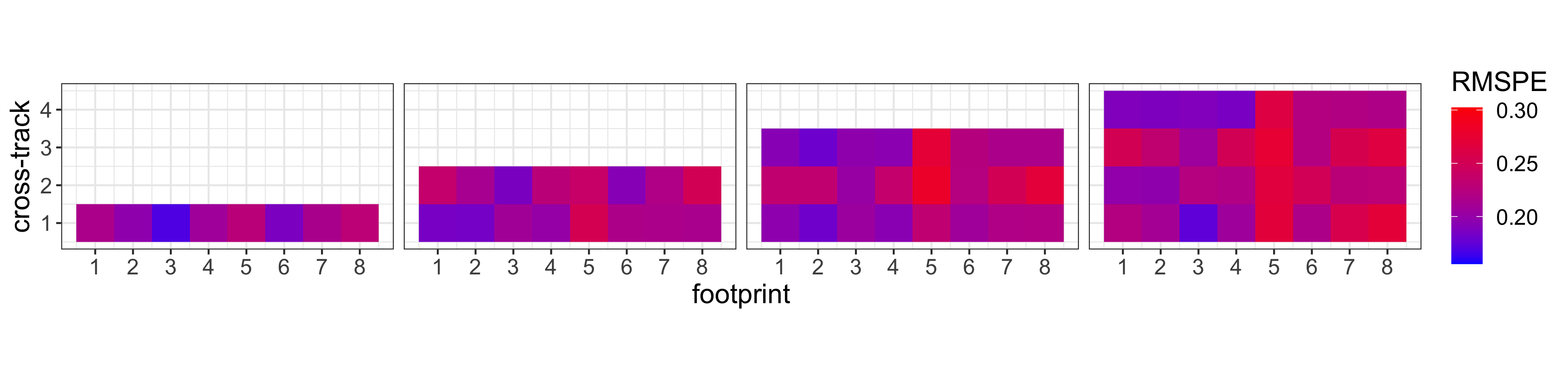}
 		\caption{}
 		\label{fig:rpmse1}
 	\end{subfigure}
 	
 	\begin{subfigure}{1\textwidth}
 		\centering
 		\includegraphics[trim = 4cm 0cm 4cm 0cm, clip = true, width = 0.95\textwidth]{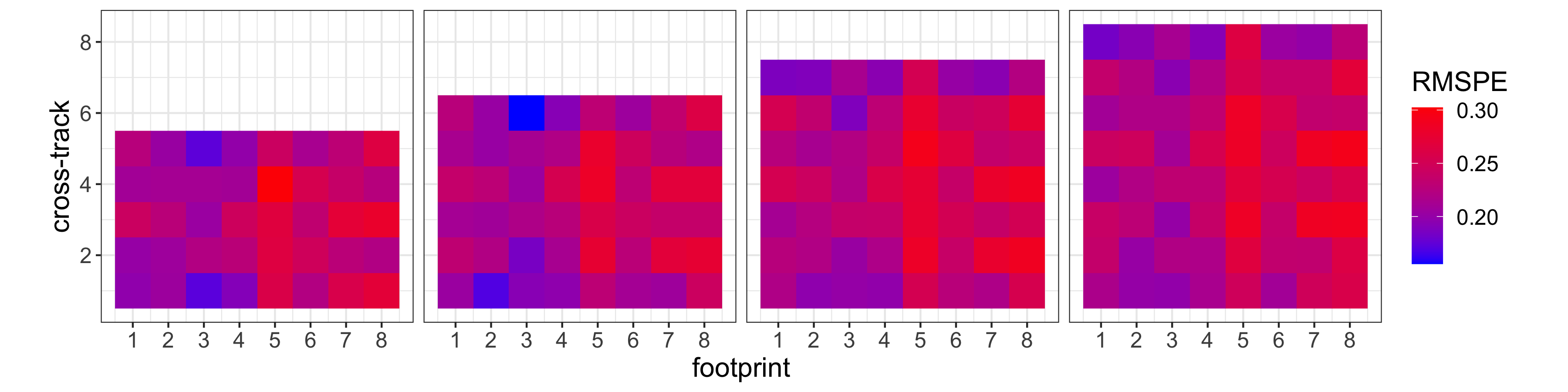}
 		\caption{}
 		\label{fig:rpmse2}
 	\end{subfigure}
 	\caption{Removed cross-tracks colored by RMSPE in radiance imputation: (a) average RMSPE in all implementations for 1-4 cross-tracks removed; (b) average RMSPE in all implementations for 5-8 cross-tracks removed.}
 \end{figure}
 
 %\spacingset{1}
 
 %\bibliographystyle{apalike}
 %\bibliography{oco2_ref.bib}

\end{document}